\documentclass[11pt,a4paper]{article}
\usepackage{amsmath,amsfonts,amssymb,amsthm}
\usepackage{amsthm}
\usepackage{graphicx,color}
\usepackage{boxedminipage}
\usepackage[ruled,boxed,linesnumbered]{algorithm2e}
\usepackage{framed}
\usepackage{thmtools}
\usepackage{thm-restate}
\usepackage{xspace}
\usepackage{float}
\usepackage{todonotes}
\usepackage{footnote}
\usepackage{mathtools}

\usepackage{vmargin}
\setmarginsrb{1in}{1in}{1in}{1in}{0pt}{0pt}{0pt}{6mm}
\usepackage{todonotes}
\usepackage{footnote}
\usepackage[pdftex, plainpages = false, pdfpagelabels, 
bookmarks=false,
bookmarksopen = true,
bookmarksnumbered = true,
breaklinks = true,
linktocpage,
pagebackref,
colorlinks = true,  
linkcolor = blue,
urlcolor  = blue,
citecolor = red,
anchorcolor = green,
hyperindex = true,
hyperfigures
]{hyperref} 
\usepackage{xifthen}
\usepackage{tabularx}
\usepackage{tikz} 
\usetikzlibrary{calc}
\usepackage{thm-autoref}
\usepackage[nameinlink]{cleveref}

\newtheorem{theorem}{Theorem}
\newtheorem{lemma}{Lemma}

\newtheorem{definition}{Definition}
\newtheorem{observation}{Observation}
\newtheorem{proposition}{Proposition}

\theoremstyle{definition}

\newcommand{\poly}{{\sf poly}}

\DeclareMathOperator{\operatorClassNP}{NP}
\newcommand{\classNP}{\ensuremath{\operatorClassNP}}

\DeclareMathOperator{\operatorClassCoNP}{coNP}
\newcommand{\classCoNP}{\ensuremath{\operatorClassCoNP}}
\DeclareMathOperator{\operatorClassFPT}{FPT\xspace}
\newcommand{\classFPT}{\ensuremath{\operatorClassFPT}\xspace}
\DeclareMathOperator{\operatorClassW}{W}
\newcommand{\classW}[1]{\ensuremath{\operatorClassW[#1]}}

\newlength{\RoundedBoxWidth}
\newsavebox{\GrayRoundedBox}
\newenvironment{GrayBox}[1]%
{\setlength{\RoundedBoxWidth}{.93\textwidth}
	\def\boxheading{#1}
	\begin{lrbox}{\GrayRoundedBox}
		\begin{minipage}{\RoundedBoxWidth}}%
		{   \end{minipage}
	\end{lrbox}
	\begin{center}
		\begin{tikzpicture}%
			\node(Text)[draw=black!20,fill=white,rounded corners,%
			inner sep=2ex,text width=\RoundedBoxWidth]%
			{\usebox{\GrayRoundedBox}};
			\coordinate(x) at (current bounding box.north west);
			\node [draw=white,rectangle,inner sep=3pt,anchor=north west,fill=white] 
			at ($(x)+(6pt,.75em)$) {\boxheading};
		\end{tikzpicture}
\end{center}}

\newenvironment{defproblemx}[2][]{\noindent\ignorespaces%
\FrameSep=6pt%
\parindent=0pt%
\vspace*{-1.5em}
\ifthenelse{\isempty{#1}}{%
\begin{GrayBox}{\textsc{#2}}%
}{%
	\begin{GrayBox}{\textsc{#2} parameterized by~{#1}}%
	}
	\begin{tabular*}{\textwidth}{@{\hspace{.1em}} >{\itshape} p{1.8cm} p{0.8\textwidth} @{}}%
	}{
	\end{tabular*}%
\end{GrayBox}%
\ignorespacesafterend
}

\newcommand{\defproblema}[3]{
\begin{defproblemx}{#1}
	Input:  & #2 \\
	Task: & #3
\end{defproblemx}
}%

\newcommand{\Oh}{\mathcal{O}}

\newcommand{\lr}[1]{\left( #1\right)}

\newcommand{\cI}{\mathcal{I}}
\newcommand{\cP}{\mathcal{P}}
\newcommand{\cS}{\mathcal{S}}

\newcommand{\pname}{\textsc}
\newcommand{\ProblemFormat}[1]{\pname{#1}}
\newcommand{\ProblemIndex}[1]{\index{problem!\ProblemFormat{#1}}}
\newcommand{\ProblemName}[1]{\ProblemFormat{#1}\ProblemIndex{#1}{}\xspace}

\newif\iflong
\longtrue 

\newcommand{\probCPack}{\ProblemName{Disk Appending}}
\newcommand{\probRearr}{\ProblemName{Disk Dispersal}}
\newcommand{\probRRearr}{\ProblemName{Rectilinear Disk Dispersal}}
\newcommand{\probGridTiling}{\ProblemName{Grid Tiling}}

\title{Kernelization for Spreading Points
	\thanks{The research leading to these results has received funding from the Research Council of Norway via the project  BWCA (grant no. 314528), the European Research Council (ERC) via grant LOPPRE, reference 819416, and the Israel Science Foundation (ISF) grant no.~1176/18.}
}

\author{
	Fedor V. Fomin\thanks{
		Department of Informatics, University of Bergen, Norway.}
	\and
	Petr A. Golovach\addtocounter{footnote}{-1}\footnotemark{}
	\and
	Tanmay Inamdar\addtocounter{footnote}{-1}\footnotemark{}
	\and
	Saket Saurabh\addtocounter{footnote}{-1}\footnotemark{} \thanks{The Institute of Mathematical Science, HBNI, Chennai, India}
	\and
	Meirav Zehavi~\thanks{Ben-Gurion University of the Negev, Beer-Sheva, Israel} 
}
\date{}

\begin{document}

\maketitle

\begin{abstract}
We consider the following problem about dispersing points. Given a set of points in the plane, the task is to identify whether by moving a small number of points by small distance, we can obtain an arrangement of points such that no pair of points is ``close'' to each other. More precisely, for a family of $n$ points,  an integer $k$, and a real number $d>0$, we ask whether at most $k$ points could be relocated, each point at distance at most $d$ from its original location, such that the distance between each pair of points is at least a fixed constant, say $1$. A number of approximation algorithms for variants of this problem, under different names like distant representatives, disk dispersing, or point spreading, are known in the literature. However, to the best of our knowledge,  the parameterized complexity of this problem remains widely unexplored. We make the first step in this direction by providing a kernelization algorithm that, in polynomial time, produces an equivalent instance with $ \Oh(d^2k^3)$ points. As a byproduct of this result, we also design a non-trivial fixed-parameter tractable (FPT) algorithm for the problem, parameterized by $k$ and $d$. Finally, we complement the result about polynomial kernelization by showing a lower bound that rules out the existence of a kernel whose size is polynomial in $k$ alone, unless $ \classNP\subseteq \classCoNP/\poly$. 

\end{abstract}

\section{Introduction}\label{sec:intro}
The problem of dispersing a family of objects is a common theme in many situations in computational geometry. It appears naturally in the wide range of settings that require assigning elements to locations. In many scenarios, dispersing has two often contradicting objectives. On the one hand, it is desirable not to place the objects too close to each other. This can be due to a variety of reasons, e.g., placing customers in a restaurant in socially distant manner, to placing wireless sensors far from each other in order to avoid interference. On the other hand, we may already have an existing placement of the objects, and wish to optimize the resources spent on moving the objects. 

With this motivation, we consider the following mathematical model of the dispersing problems. In this model, our aim is to modify a given arrangement of points in the plane, by moving some of the points into new positions within a given distance, such that the Euclidean distance between each pair of points in the final arrangement is at least a fixed constant, say $2$. Equivalently, the problem can be reformulated in terms of finding a non-overlapping arrangement of unit disks, formulated below as the problem \probRearr. 

\defproblema{\probRearr}%
{A family $\cS$ of $n$ unit disks, an integer $k\geq 0$, and a real $d\geq 0$.}%
{Decide whether it is possible to obtain  from $\cS$ a family of non-overlapping unit disks $\cP$ by moving at most $k$ disks into new positions in such a way that each unit disk is moved a distance at most $d$. \footnotemark
}
\footnotetext{All (unit) disks considered in the paper are open unless specified otherwise. In particular, two unit disks touching each other are not considered to be overlapping. Due to this simplifying assumption, we avoid the discussion about placing disks such that the distance between their boundaries is infinitesimally small.}

\probRearr---and therefore, the problem of spreading points---is closely related to the problem of finding a system of \emph{$q$-distant representatives}. This problem was introduced by Fiala, Kratochv{\'{\i}}l, and  Proskurowski
\cite{FialaKP05} as a geometric extension of the  classic combinatorial notion of the ‘‘systems of distinct representatives’’. For a set of geometric objects in a metric space and a number $q>0$, the task is to choose one representative point from each object such that the selected points are at a distance at least $q$  from each other. For $k=n$, an instance $({\cal S},d,k)$ of \probRearr\ can be viewed as an instance of the problem of finding  a system of  $q$-distance representatives by setting $q=2$ and defining the set of geometric objects as follows: for each disk $D\in{\cal S}$, create a disk with the same center but with radius $d$ (instead of $1$). 
This yields that \probRearr is also \classNP-hard for $d=2$ from the result of \cite{FialaKP05}.

The problem  of computing the distant representatives  has applications in map labeling and data vizualization, where the goal is to   place  labels as close as possible to the specified features of the map but avoiding overlapping  (thus the centers of labels are the centers of non-intersecting disks, ensuring that they are sufficiently separated) \cite{DoddiMMMZ97,Jiang06a,JiangQQZC03}.
The problem is also related to 
problems of ``imprecise points'' \cite{LofflerK10,SheikhiMBM17}, the settings  where
locations of points are given with some precision. Approximation algorithms for this and related point spreading problems---where the goal is to place the specified number of points within a certain region so as to maximize the smallest pairwise distance between the points---were developed in  
\cite{biedlL21,Cabello07,DemaineHMSGZ09,DumitrescuJ10,DumitrescuJ11,DumitrescuJ12,fekete2004maximum, hsiang2004algorithms,baur2001approximation, hsiang2003online}.  


\begin{figure} 
	\begin{center}
		\includegraphics[scale=2]{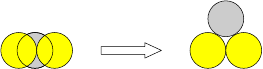}%
		\caption{An example of \probRearr with $k=1$ and $d=\sqrt{3}$.  A non-overlapping arrangement of disks obtained from a family of three disks by moving the central disk at distance $\sqrt{3}$.}\label{fig:introex2}
	\end{center}
\end{figure}

To the best of our knowledge,  the parameterized complexity of dispersal problems are widely unexplored. The notable exception is the work of 
Demaine, Hajiaghayi, and  Marx
\cite{DemaineHM14} on dispersion in graphs. In this problem, we are given an underlying edge-weighted graph, called the \emph{connectivity graph} $G$, and a set of $k$ ``agents'' or ``pebbles'', located at a subset of vertices $G$. The task is to move the pebbles to distinct vertices and such that no two pebbles are adjacent.  The movement problem is  \classW1-hard parameterized by the number of pebbles, even in the case when each pebble is allowed to move at most one step.

\subsection{Our Results}
Our first result concerns kernelization (polynomial compression) of \probRearr. Informally speaking, in parameterized complexity, the polynomial kernel is a polynomial-time algorithm that compresses the instance of a parameterized problem to the instance whose size is bounded by a polynomial of the parameter. \Cref{thm:partialKernel}  gives an algorithm that runs in polynomial time, and reduces the number of disks to some polynomial of $d$ and $k$. 



\begin{restatable}{theorem}{themainkernel}
	\label{thm:partialKernel} 
	There is a polynomial-time algorithm that, given an instance $(\cS,k,d)$ of  \probRearr, outputs an equivalent instance $(\cS',k,d)$ of the same problem, where   the number of unit disks is $|\cS'|=\Oh((d+1)^2k^3)$, and $\cS'\subseteq \cS$.
\end{restatable}

Strictly speaking, the algorithm in \Cref{thm:partialKernel} is not a polynomial kernel according to the standard definition of this notion---we do not guarantee that the coordinates of disks, and thus the overall size of the compressed instance, is bounded by a polynomial in $k$ and $d$. We call such a compression algorithm a \emph{partial kernel}. Further, we observe in \Cref{thm:kern-denom} that the partial kernel from \Cref{thm:partialKernel} can be modified to be a polynomial kernel if the centers of input disks are constrained to be rationals and we parameterize the problem by $k$, $d$, and the maximum denominator of coordinates of centers. 

For a parameterized problem, given the existence of a (partial) kernel, it is usually straightforward to design a fixed-paramter tractable (FPT) algorithm by an exhaustive enumeration of all candidate solutions. For \probRearr, however, this is not entirely obvious. After computing an equivalent reduced instance by applying \Cref{thm:partialKernel}, one can enumerate all possible subsets of at most $k$ unit disks that are to be moved. Now, for each such subset, we want to decide whether each unit disk in the subset can be moved by a distance of at most $d$ that results in a non-overlapping configuration. Since there are infinitely many possible target locations for each unit disk, this step requires some additional work. We show that this decision subroutine can be reduced to checking whether a system of polynomial inequalities has a solution over real numbers, which can then be determined in FPT time by using classical results from computational real algebra. Thus, we obtain the following non-trivial corollary.

\begin{restatable}{corollary}{themainFPT} \label{cor:fpt}
	\probRearr\ is \classFPT when parameterized by $d+k$. Specifically,
	it is solvable in time $(dk)^{\Oh(k)} \cdot |I|^{\Oh(1)}$.
	\label{mainFPT}
\end{restatable}

Our next result is a companion lower bound to the partial kernelization of  \Cref{thm:partialKernel}, which shows that one cannot remove the dependence on $d$ from the kernel size.

\begin{restatable}{theorem}{theoremnopoly}
	\label{thm:nopolyker}
	\probRearr parameterized by $k$ does not admit a polynomial kernel unless $\classCoNP \subseteq \classNP/\poly$. This result holds even if the distance $d$ is an integer, and the centers of the given disks have rational coordinates.
\end{restatable}

As we already mentioned, by the result of Fiala, Kratochv{\'{\i}}l, and  Proskurowski about {$q$-distant representatives}, \probRearr is  \classNP-hard for $d=2$. Thus the problem is in the class para-NP for parameter $d$. However, the complexity of parameterization by $k$ is more interesting, which remains open. However, in \Cref{sec:rearr-w-hard}, we make a preliminary progress on this question, by showing that the rectilinear version of \probRearr, called \probRRearr, is indeed \classW{1}-hard parameterized by $k$. This problem is defined as follows. 

\defproblema{\probRRearr}%
{A family $\cS$ of $n$ unit disks, an integer $k\geq 0$, and a real $d\geq 0$.}%
{Decide whether it is possible to obtain  from $\cS$ a family of non-overlapping disks  $\cP$ by moving at most $k$ disks into new positions parallel to the axes in such a way that each disk is moved at distance at most $d$.}

More formally, we show the following result regarding \probRRearr.

\begin{restatable}{theorem}{themaiWhards}
	\label{thm:w1hardness}
	\probRRearr
	is $\classW{1}$-hard parameterized by $k$, number of disks that are allowed to be moved. This result holds even if centers of the given disks lie on an integer grid, and the value of $d$, the maximum distance by which a disk can be moved, is an integer.\label{mainWh}
\end{restatable}

\paragraph*{Organization.} In \Cref{sec:prelim} we introduce basic notions. In \Cref{sec:rearr-kern}, we consider kernelization for \probRearr. Further, we give complexity lower bounds. In  \Cref{sec:rearr-nokern}, we show that it is unlikely that \probRearr admits a polynomial kernel when parameterized by $k$ only. In \Cref{sec:rearr-w-hard}, we show that \probRRearr is $\classW{1}$-hard parameterized by $k$, proving \Cref{thm:w1hardness}. Finally, in \Cref{sec:concl}, we provide some concluding remarks and future directions.

\section{Preliminaries}\label{sec:prelim} 

As it is common in computational geometry, we assume the \emph{real RAM} computational model, that is, we are working with real numbers and assume that basic operations can be executed in unit time. 

\paragraph*{Disks and Segments.} For two points $A$ and $B$ in the plane, we use $AB$ to denote the line segment with endpoints at $A$ and $B$. The \emph{distance} between 
$A=(x_1,y_1)$ and $B=(x_2,y_2)$  or the \emph{length} of $AB$, is $|AB|=\|A-B\|_2=\sqrt{(x_1-x_2)^2+(y_1-y_2)^2}$. 
The \emph{(open unit) disk} with a \emph{center} $C=(c_1,c_2)$ in the plane is the set of points $(x,y)$ satisfying the inequality $(x-c_1)^2+(y-c_2)^2<1$.  
Whenever we write ``disk'' we mean an \emph{open unit disk}, unless radius or closed-ness is specified explicitly. 
Clearly, two disks with centers $A$ and $B$ are disjoint if and only if the distance between $A$ and $B$ is at least two.  We say that the disks \emph{touch} if $|AB|=2$.
For real numbers $a
\leq b$, we use $[a,b]=\{x\in\mathbb{R}\mid a\leq x\leq b\}$ to denote a \emph{closed interval}. For $a_1\leq b_1$ and $a_2\leq b_2$, $[a_1,b_1]\times [a_2,b_2]=\{(x,y)\in\mathbb{R}^2\mid a_1\leq x\leq b_1\text{ and }a_2\leq y\leq b_2\}$. 
A point $X$ is \emph{properly inside} of a polygon $P$ if it is inside $P$ but $X$ is not on the boundary; if we say that $X$ is inside $P$, we allow it to be on the boundary.    
A disk is \emph{(properly) iniside} of a polygon $ P$ if every point of the disk is (properly) inside of $P$.  

\paragraph*{Graphs.} We use standard graph-theoretic terminology and refer to the textbook of Diestel~\cite{Diestel12} for definitions of standard notions. Let $\cS$ be a set of geometric objects in the plane (i.e., non-empty subsets of $\mathbb{R}^2$). Then, it is possible to define an intersection graph $G(\cS)$ as follows: $G(\cS)$ contains a unique vertex corresponding to every object in $\cS$, and there is an edge between the two vertices iff the corresponding two objects in $\cS$ have a non-empty intersection. 
\emph{Unit disk graphs} are the intersection graphs of unit disks in the plane. Note that, given a family $\cS$ of unit disks, we can construct the corresponding unit disk graph $G(\cS)$ in quadratic time. 

\paragraph{Parameterized Complexity.} We refer to the standard textbooks (\cite{CyganFKLMPPS15,FominLSZ19}) for introduction to the area and formal definitions. Here, we only give a brief overview. Let $(\cI, k)$ be an instance of a decision problem $\Pi$, where $k$ is a non-negative integer. We say that $\Pi$ is \emph{fixed-parameter tractable} by $k$, if there exists an algorithm that can decide whether $\cI$ is a yes-instance of $\Pi$ in time $f(k) \cdot |\cI|^{\Oh(1)}$ for some computable function $f$, where $|\cI|$ denotes the size of the instance $\cI$. A common way to show that it is unlikely that a parameterized problem is in \classFPT, one can prove that it is \classW{1}-hard by demonstrating a  \emph{parameterized reduction} from a known \classW{1}-hard problem; we refer to \cite{CyganFKLMPPS15} for the formal definitions of the class \classW{1} and parameterized reductions. 

A \emph{kernelization} (or \emph{kernel}) for $\Pi$ is a polynomial time algorithm that, given an instance $(\cI,k)$ of $\Pi$, outputs an equivalent instance $(\cI',k')$ of $\Pi$ such that $|\cI'|+k'\leq g(k)$ for a computable function $g$.  A kernel is \emph{polynomial} if $g$ is a polynomial. It can be shown that every decidable \classFPT problem admits a kernel. However, it is unlikely that all \classFPT problems have polynomial kernels. In particular, there is the now standard \emph{cross-composition} technique to show that a parameterized problem does not admit a polynomial kernel unless $\classNP\subseteq \classCoNP/\poly$.

\paragraph*{Systems of Polynomial Inequalities.} In our FPT algorithm, we will need to find suitable locations for new disks that need to be added such that the locations are ``compatible'' with an existing arrangement of disks. We will achieve this by solving systems of polynomial inequalities. We use the following result.
\begin{proposition}[Theorem 13.13 in \cite{basu06}] \label{prop:polyequations}
	Let $R$ be a real closed field, and let $\mathcal{P} \subseteq R[X_1, \ldots, X_k]$ be a finite set of $s$ polynomials, each of degree at most $c$, and let \[(\exists X_1) (\exists X_2) \ldots (\exists X_k) F(X_1, X_2, \ldots, X_k)\]
	be a sentence, where $F(X_1, \ldots, X_k)$ is a quantifier-free boolean formula involving $\mathcal{P}$-atoms of type $P \odot 0$, where $\odot \in \{ =, \neq, > , < \}$, and $P$ is a polynomial in $\mathcal{P}$. Then, there exists an algorithm to decide the truth of the sentence with complexity $s^{k+1} c^{\Oh(k)}$ in $D$,\footnote{That is, the algorithm performs $s^{k+1} c^{\Oh(k)}$ operations in $D$.} where $D$ is the ring generated by the coefficients of the polynomials in $\mathcal{P}$. 
\end{proposition}

Furthermore, a point $(X_1^*,\ldots,X_k^*)$ satisfying $F(X_1, \ldots, X_k)$ can be computed in the same time by Algorithm~13.2 (sampling algorithm) of~\cite{basu06} (see Theorem~13.11 of \cite{basu06}).

\section{Kernelization and FPT Algorithms for  \probRearr}\label{sec:rearr-kern} 
In this section, we first prove \Cref{thm:partialKernel} on partial kernel for \probRearr parameterized by $k+d$. Specifically, the output instance of the partial kernel is guaranteed to consist of only $\Oh(d^2k^3)$ unit disks. In case the coordinates of the disks in the input instance are rationals of the form $a+\frac{b}{c}$ where $b,c$ are bounded by a fixed constant (or a polynomial in $k+d$), our partial kernel in fact yields a (normal) kernel. Finally, using our partial kernel, we prove in \Cref{mainFPT} that \probRearr is FPT parameterized by $k+d$.

\begin{figure}
	\begin{center}
		\includegraphics[scale=0.75]{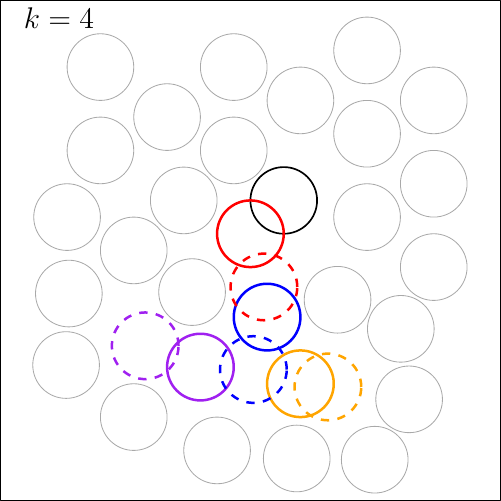}
		\caption{Example of the propagation effect. The dotted objects correspond to a solution where an object of a certain color is replaced by the dashed object of the same color.}
		\label{fig:introPropagation}
	\end{center}
\end{figure}

The proofs of our partial kernels begin with the simple observation that if we are given a yes-instance, then the unit disk graph corresponding to the input set of unit disks admits a vertex cover of size at most $k$. So, in polynomial time we obtain a vertex cover $U$ of size at most $2k$. At first glance, one may think to remove all input unit disks that do not intersect any unit disk in $U$. However, we might  be forced to perform movement operations that make some neighborhood sets larger (e.g., see \Cref{fig:introPropagation}), which, in turn, can have a propagating effect that forces us to move unit disks that are ``quite far'' from all unit disks in $U$. Still, we can prove by induction on $k$ that if the input instance is a yes-instance, then it admits a solution where all the unit disks that are moved are at distance at most $\Oh(d^2k^2)$ from at least one unit disk in $U$. This gives rise to a reduction rule where we only keep the unit disks within this distance from at least one unit disk in $U$ as well as additional unit disks at some (almost negligible) distance from them. 

After having reduced the number of unit disks, we can shift the unit disks that we keep so that the coordinates of their centers will be polynomial in $k+d$, under the assumption that the coordinates of the unit disks in the input instance are rationals of the form $a+\frac{b}{c}$ where $b,c$ are bounded by a fixed constant (or a polynomial in $k+d$). To obtain FPT algorithms, we first apply our partial kernels. Afterwards, we guess which disks to move. Then, we determine how to move them by solving a corresponding system of polynomial inequalities.

For the sake of formality, we will use the notion of a solution in this section as follows.

\begin{definition}
	Let $({\cal S},k,d)$ be an instance of \probRearr. A {\em solution} is a bijective function $\mathsf{move}\colon {\cal S}\rightarrow{\cal P}$ such that:
	\begin{enumerate}
		\item  ${\cal P}$ is a \emph{packing}, i.e., a non-overlapping set of unit disks.
		\item $|\{D\in{\cal S}: \mathsf{move}(D)\neq D\}|\leq k$.
		\item For every $D\in{\cal S}$: The distance between the centers of $D$ and $\mathsf{move}(D)$ is at most $d$.
	\end{enumerate}
	We define the {\em set of unit disks moved by $\mathsf{move}$} as $\{D\in{\cal S}: \mathsf{move}(D)\neq D\}$, and the {\em size of $\mathsf{move}$} as the size of this set.
\end{definition}

Notice that any set of unit disks that is moved by a solution to \probRearr is in particular a vertex cover (though not necessarily a minimal one) for the  intersection graph of the input set of unit disks. As previously discussed, since the {\sc Vertex Cover} problem admits a 2-approximation algorithm in polynomial time, this yields the following observation.

\begin{observation}\label{obs:vcApprox}
	There exists a polynomial-time algorithm that, given an instance $({\cal S},k,d)$ of \probRearr, either correctly concludes that $({\cal S},k,d)$ is a no-instance, or outputs a vertex cover of size at most $2k$ for the unit disk graph corresponding to $\cal S$.
\end{observation}

We will also need the following observation, which is directly implied by the fact that the area of a disk of radius $r$ is $\pi r^2$, while the area of a unit disk (whose radius is $1$) is $\pi$.

\begin{observation}\label{obs:packingDisks}
	The number of pairwise non-intersecting unit disks in a disk of radius $r$ is at most $r^2$.
\end{observation}

Towards the presentation of our partial kernel, we need to prove one lemma. Informally speaking, this lemma shows that the set of disks that may be potentially moved in a yes-instance is contained in a bounded area around a small number of disks, in particular the disks that form a vertex cover in the intersection graph. Furthermore, since all such disks, except that forming the vertex cover, are non-intersecting, this lemma eventually helps us bound the number of such disks by a polynomial in $k$ and $d$.
\begin{lemma}\label{lem:limitEffectShifting}
	Let $({\cal S},d,k)$ be a yes-instance of \probRearr. Let $U$ be a vertex cover for the intersection graph of $\cal S$. Then, any minimum-sized solution to $({\cal S},k,d)$ only moves unit disks whose center is at distance at most $(d+2)\cdot k$ from the center of at least one unit disk in $U$.
\end{lemma}

\begin{proof}
	We prove the lemma by induction on $k$. When $k=0$, the only minimum-sized solution to $({\cal S},k,d)$ is the one that moves no unit disk, and hence the claim trivially follows. Now, suppose that the claim holds for $k-1\geq 0$, and let us prove it for $k$.  If the intersection graph of $\cal S$ is edgeless, then the only minimum-sized solution to $({\cal S},k,d)$ is the one that moves no unit disk, and hence the claim trivially follows as in the base case. So, we can next suppose that there exist two different unit disks $D,D'\in{\cal S}$ that intersect each other. See \Cref{fig:lemshifting} for an illustration.
	
	Since $U$ is a vertex cover, it must contain at least one unit disk among $D$ and $D'$, denoted by $X$.  Moreover, any solution to $({\cal S},k,d)$ must move at least one unit disk among $D$ and $D'$. Let $\mathsf{move}:{\cal S}\rightarrow{\cal P}$ be an arbitrary minimum-sized solution to $\cI= ({\cal S},k,d)$, and let $Y$ be a unit disk among $D$ and $D'$ that $\mathsf{move}$ moves to attain ${\cal P}$. Let $Y'=\mathsf{move}(Y)$, and let ${\cal S}'=({\cal S}\setminus\{Y\})\cup\{Y'\}$. We attain solution $\mathsf{move}':{\cal S}'\rightarrow{\cal P}'$ to a new instance $\cI' = ({\cal S}',k-1,d)$ as follows: for every $\tilde{D}\in{\cal S}\setminus\{Y\}, \mathsf{move}'(\tilde{D})=\mathsf{move}(\tilde{D})$; $\mathsf{move}'(Y')=Y'$. Note that $\mathsf{move}'$ must be a minimum-sized solution to $({\cal S}', k-1, d)$, otherwise we can obtain a solution for the original instance $({\cal S}, k, d)$ that is smaller than $\mathsf{move}$, contradicting its optimality. Further, note that $(U\setminus \{Y\})\cup\{Y'\}$ is a (not necessarily minimal)  vertex cover for the intersection graph of ${\cal S}'$. By the inductive hypothesis, this means that $\mathsf{move}'$  only moves unit disks whose center is at distance at most $(d+2)\cdot (k-1)$ from the center of at least one unit disk in $(U\setminus \{Y\})\cup\{Y'\}$. Moreover, the distance between the centers of $Y$ and $X$ is at most $2$ (since they intersect) and the distance between the centers of $Y'$ and $Y$ is at most $d$, so the distance between the centers $Y'$ and $X$ is at most $d+2$. In turn, this means that $\mathsf{move}$ only moves unit disks at distance at most $(d+2)\cdot k$ from at least one unit disk in $U$, which concludes the proof.
\end{proof}

\begin{figure}
	\centering
	\includegraphics[scale=0.85]{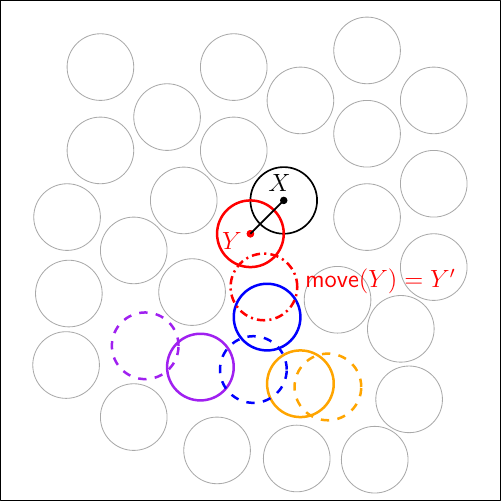}
	\caption{Illustration for Proof of \Cref{lem:limitEffectShifting}. A vertex cover $U$ contains disk $X$ and a solution $\cal S$ moves a disk $Y$ to its new location, $Y' = \mathsf{move}(Y)$, denoted in dash-dotted disk in red color. A new instance $\cI'$ is obtained by replacing $Y$ with $Y'$ and reducing the budget by $1$, and $U' = X \cup Y'$ is a vertex cover for the resulting intersection graph. A solution to $\cI'$ moves the solid blue, purple, and orange disks to their new locations, shown in dashed disks of corresponding color. By inductive hypothesis, the new locations are at distance at most $(d+2) \cdot (k-1)$ from $U'$, and the distance between $X$ and $Y'$ is at most $d+2$.}
	\label{fig:lemshifting}
\end{figure}

We are now ready to present the partial kernel for \probRearr. For the reader's convenience, we restate  \Cref{thm:partialKernel} here. 

\themainkernel*

\begin{proof}
	Given an instance $({\cal S},k,d)$ of  \probRearr, the (partial kernel) kernelization algorithm works as follows. Based on \Cref{obs:vcApprox}, it computes a vertex cover $U$  of size at most $2k$ for the intersection graph of $\cal S$. Then, it obtains ${\cal S}'$ from $\cal S$ by removing from $\cal S$ all the unit disks at distance more than $(d+2)\cdot (k+1)$ from all unit disks in $U$. The output instance is $({\cal S}',k,d)$. Clearly, the kernelization algorithm works in polynomial time. So, it suffices to prove that $({\cal S},k,d)$ and $({\cal S}',k,d)$ are equivalent and that $|{\cal S}'|=\Oh(d^2k^3)$.
	
	We first prove the equivalence. In one direction, suppose that $({\cal S},k,d)$ is a yes-instance, and let $\mathsf{move}:{\cal S}\rightarrow{\cal P}$ be a solution to it. In particular, the restriction of $\mathsf{move}$ to ${\cal S}'$ clearly yields a packing (being a subset of $\cal P$) and moves at most as many disks as $\mathsf{move}$ does. So, the restriction of $\mathsf{move}$ to ${\cal S}'$ is a solution to $({\cal S}',k,d)$.
	
	In the other direction, suppose that $({\cal S}',k,d)$ is a yes-instance. By \Cref{lem:limitEffectShifting}, $({\cal S}',k,d)$ admits a solution $\mathsf{move}':{\cal S}'\rightarrow{\cal P}'$ that only moves unit disks whose centers are at distance at most $(d+2)\cdot k$ from the center of at least one unit disk in $U$.\footnote{Note that ${\cal S}'$ may contains unit disks whose centers are at distance larger than $(d+2)\cdot k$ (but at most $(d+2)\cdot (k+1)$) from the centers of all unit disks in $U$.} Define $\mathsf{move}:{\cal S}\rightarrow{\cal P}$ for ${\cal P}={\cal P}'\cup({\cal S}\setminus{\cal S}')$ as follows: for every $D\in{\cal S}'$, $\mathsf{move}(D)=\mathsf{move}'(D)$, and for every $D\in{\cal S}\setminus{\cal S}'$, $\mathsf{move}(D)=D$. We claim that $\mathsf{move}$ is a solution to $({\cal S},k,d)$. To this end, first note that none of the unit disks in ${\cal P}'$ intersect each other (since $\mathsf{move}'$ is a solution to $({\cal S}',k,d)$). In particular, the unit disks in $\{D\in U: \mathsf{move}(D)=D\}$ do not intersect any other unit disk in ${\cal P}'$. However, all unit disks in $\cal S$ that do not belong to $U$ do not intersect each other (since $U$ is a vertex cover for the intersection graph of $\cal S$). So, in $\cal P$, the only pairs of unit disks that can potentially intersect each other are pairs where one is a unit disk that was moved by $\mathsf{move}$ and the other belongs to ${\cal S}\setminus{\cal S}'$. However, the center of any unit disk $D$ that is moved by $\mathsf{move}$ is at distance at most $(d+2)\cdot k$ from the center of at least one unit disk $D'$ in $U$, and hence the center of $\mathsf{move}(D)$ is at distance at most $d+(d+2)\cdot k$ from the center of $D'$, while the center of any unit disk in ${\cal S}\setminus{\cal S}'$ is at distance more than $(d+2)\cdot (k+1)$ from the centers of all unit disks in $U$. Thus, $\cal P$ cannot have a pair of unit disks that intersect each other, such that one is a unit disk that was moved by $\mathsf{move}$ and the other belongs to ${\cal S}\setminus{\cal S}'$. So, $\mathsf{move}$ is indeed a solution to $({\cal S},k,d)$.
	
	Now, note that for every $D\in U$, the unit disks whose center is at distance at most $(d+2)\cdot (k+1)$ from $D$ are contained in a disk $D'$ of radius $(d+2)\cdot (k+1)+1$ and whose center is the same as the center of $D$. So, by \Cref{obs:packingDisks} and since $U$ is a vertex cover for the intersection graph of $\cal S$, this means that there exist at most $((d+2)\cdot (k+1)+2)^2=\Oh(d^2k^2)$ unit disks in ${\cal S}\setminus U$ that intersect $D'$. As $|U|\leq 2k$, we conclude that $|{\cal S}'|\leq |U| + |U|\cdot ((d+2)\cdot (k+1)+2)^2=\Oh(d^2k^3)$.
\end{proof}

To reduce the bitsize of encoding the coordinates of the unit disks in the output instance, we make use of the following lemma.

\begin{lemma}\label{lem:shrinkCoordinatesHelper}
	There exists a polynomial-time algorithm that, given a set ${\cal D}$ of unit disks whose centers have rational coordinates, a partition $({\cal D}_1,{\cal D}_2,\ldots,{\cal D}_\ell)$ of ${\cal D}$, and $r\in\mathbb{N}$, outputs a set ${\cal D}'$ of unit disks whose centers have rational coordinates and a bijective function $f:{\cal D}\rightarrow{\cal D}'$ with the following properties.
	\begin{itemize}
		\item For all $i\in\{1,2,\ldots,\ell\}$, ${\cal D}_i$ and $\{f(D): D\in{\cal D}_i\}$ are \emph{isometric}, that is, for all $D,D'\in{\cal D}_i$, we have $\mathsf{distance}(D,D')=\mathsf{distance}(f(D),f(D'))$.
		\item For all distinct $i,j\in\{1,2,\ldots,\ell\}, D\in{\cal D}_i$ and $D'\in{\cal D}_j$, we have $\mathsf{distance}(D,D')> r$.
		\item Encoding the coordinates (in unary) of all the unit disks in $\{f(D): D\in{\cal D}_i\}$ requires space polynomial in $\displaystyle r,
		|{\cal D}|,m=\max_{i=1}^\ell\max_{D,D'\in{\cal D}_i}\mathsf{distance}(D,D')$ and $N=\max_{b,c}(b+c)$ over every $b,c\in\mathbb{N}, b<c,$ and $b, c$ are coprime, such that $a+\frac{b}{c}$ is a coordinate of a center of a unit disk in $\cal D$.
	\end{itemize}
\end{lemma}

\begin{proof}
	For every $i\in\{1,2,\ldots,\ell\}$, let $L_i$ be a leftmost unit disk in ${\cal D}_i$ (i.e., with a smallest $x$-coordinate of its center), and let $D_i$ be a bottommost unit disk in ${\cal D}_i$ (i.e., with a smallest $y$-coordinate of its center), and denote their centers by $(x^{\mathsf{left}}_i,y^{\mathsf{left}}_i)$ and $(x^{\mathsf{bottom}}_i,y^{\mathsf{bottom}}_i)$, respectively.  Now, for every $i\in\{1,2,\ldots,\ell\}$ and every $D\in{\cal D}_i$ with center $(x,y)$, define $f(D)$ as the unit disk whose center is $(x-x^{\mathsf{left}}_i + (i-1)\cdot (m+r), y-y^{\mathsf{bottom}}_i + (i-1)\cdot (m+r))$. We define ${\cal D}'$ as the set of unit disks assigned by $f$. Clearly, $f:{\cal D}\rightarrow{\cal D}'$ is bijective and the third property in the lemma holds.
	
	For the first property, consider two unit disks $D,D'\in {\cal D}_i$ for some $i\in\{1,2,\ldots,\ell\}$ with centers $(x,y)$ and $(x',y')$, respectively. Then, $\mathsf{distance}(f(D),f(D'))$ is equal to the square root of $((x-x^{\mathsf{left}}_i + (i-1)\cdot (m+r))-(x'-x^{\mathsf{left}}_i + (i-1)\cdot (m+r)))^2+((y-y^{\mathsf{bottom}}_i + (i-1)\cdot (m+r))-(y'-y^{\mathsf{bottom}}_i + (i-1)\cdot (m+r)))^2$, which is precisely $\sqrt{(x-x')^2+(y-y')^2} = \mathsf{distance}(D,D')$.
	So, the first property in the lemma holds.
	
	For the second property, consider two unit disks $D\in{\cal D}_i,D'\in {\cal D}_j$ for some $i,j\in\{1,2,\ldots,\ell\}$, $i<j,$ with centers $(x,y)$ and $(x',y')$, respectively. Then, $\mathsf{distance}(f(D),f(D'))$ is equal to the square root of $((x'-x^{\mathsf{left}}_j + (j-1)\cdot (m+r))-(x-x^{\mathsf{left}}_i + (i-1)\cdot (m+r)))^2+((y'-y^{\mathsf{bottom}}_j + (j-1)\cdot (m+r))-(y-y^{\mathsf{bottom}}_i + (i-1)\cdot (m+r)))^2$. Observe that $x'\geq x^{\mathsf{left}}_j, y'\geq y^{\mathsf{bottom}}_j, x\leq x^{\mathsf{left}}_i+m, y\leq y^{\mathsf{bottom}}_i+m$. So, the above expression is lower bounded by 
	\[\begin{array}{l}
		\displaystyle{\sqrt{2\left((j-1)\cdot (m+r) - (m+ (i-1)\cdot(m+r) )\right)^2}}
		= \displaystyle{\sqrt{2}\cdot ((j-i)(m+r)-m)\geq \sqrt{2}r.}
		
	\end{array}\]
	In particular, $\mathsf{distance}(f(D),f(D'))> r$. So, the second property in the lemma holds.
\end{proof}

We will also need the following simple observation.

\begin{observation}\label{obs:limitEffect}
	Let $\cal S$ be a set of unit disks in the Euclidean plane. Let $D\in {\cal S}$. Then, by moving $D$ by a distance of at most some $d\in\mathbb{N}$,  $D$ cannot intersect unit disks whose centers  are at distance at least $d+2$ from the original position of the center of $D$.
\end{observation}

Based on \Cref{lem:shrinkCoordinatesHelper} and \Cref{obs:limitEffect}, we prove the following.

\begin{lemma}\label{lem:shrinkCoordinates}
	There exists a polynomial-time algorithm that, given an instance $({\cal S},k,d)$ of \probRearr\ where the centers of all disks have rational coordinates, and a partition $({\cal S}_1,{\cal S}_2,\ldots,{\cal S}_\ell)$ of ${\cal S}$ such that for all $i,j\in\{1,2,\ldots,\ell\},D\in{\cal S}_i$ and $D'\in{\cal S}_j$, we have $\mathsf{distance}(D,D')\geq 2d+2$, outputs an equivalent instance of \probRearr, respectively, with the same parameters $k,d$ and number of unit disks, where encoding the coordinates of all the unit disks (in unary) requires space polynomial in $d$, $|{\cal S}|$, $m=\displaystyle{\max_{i=1}^\ell\max_{D,D'\in{\cal S}_i}\mathsf{distance}(D,D')}$ and $N=\max_{b,c}(b+c )$ over every $b,c\in\mathbb{N}, b<c,$ such that $a+\frac{b}{c}$ is a coordinate of a center of a unit disk in $\cal D$.
\end{lemma}

\begin{proof}
	The algorithm simply applies the algorithm in \Cref{lem:shrinkCoordinatesHelper} with $r=2d+2$, and obtains $f:{\cal S}\rightarrow{\cal S}'$. Then, it returns ${\cal D}'$. From \Cref{lem:shrinkCoordinatesHelper}, it directly follows that encoding the coordinates of all the unit disks requires space polynomial in $d,|{\cal D}|,m$ and $N$.
	Recall that for all $i,j\in\{1,2,\ldots,\ell\},D\in{\cal S}_i$ and $D'\in{\cal S}_j$, we have $\mathsf{distance}(D,D')\geq 2d+2$, and this property is preserved under the mapping $f$ (by our choice of $r$). So, \Cref{obs:limitEffect} implies that the sub-instances induced by the different sets ${\cal S}_i$ are ``independent'' from each other: we cannot move unit disks in one set ${\cal S}_i$ so that they intersect unit disks in another set ${\cal S}_j$. Also, the same holds for the sub-instances they are mapped to by $f$. As every sub-instance induced by some set ${\cal S}_i$ is equivalent to the sub-instance it is mapped to by $f$ since ${\cal S}_i$ and $\{f(D): D\in{\cal S}_i\}$ are isometric, we conclude that $({\cal S},k,d)$ and $({\cal S}',k,d)$ are equivalent.
\end{proof}

We our now ready to present our (non-partial) kernel for \probRearr. In particular, if $N$ is a constant (or polynomial in $k+d$), the parameterization can be assumed to be only by $k+d$. 

\begin{theorem}\label{thm:kern-denom}
	\probRearr, restricted to instances where the centers of all disks have rational coordinates, admits a polynomial kernel with respect to $k+d+N$, where $N=\max_{b,c}(b+c )$ over every $b,c\in\mathbb{N}, b<c,$ such that $a+\frac{b}{c}$ is a coordinate of a center of a unit disk in $\cal S$.
\end{theorem}

\begin{proof}
	Given an instance $({\cal S},k,d)$ of \probRearr, restricted to instances where the centers of all disks have rational coordinates, the kernelization algorithm works as follows. First, we call the algorithm in \Cref{thm:partialKernel} to obtain an equivalent instance $({\cal S}',k,d)$ of \probRearr. Here, $k,d$ remain unchanged, and ${\cal S}'$ is a subset of $\cal S$. Let ${\cal W}=\{W_D: D\in {\cal S}'\}$ where $W_D$ is a disk whose center is the same as the center of $D$ and whose radius is $d+1$. Let ${\cal C}$ be the set of connected components of the intersection graph of $\cal W$. Let ${\cal P}$ be the partition of ${\cal S}'$ such that two unit disks in ${\cal S}'$ belong to the same part if and only if there exists a connected component in $\cal C$ such that both are intersected by (possibly different) disks that belong to that component. It should be clear, from the definitions of ${\cal S}'$ and ${\cal W}$, that this is indeed a partition, and that if two unit disks in ${\cal S}'$ belong to different parts in this partition, then the distance between their centers is larger than $2d+2$.
	So, the kernelization algorithm then calls the algorithm in \Cref{lem:shrinkCoordinates} on $({\cal S}',k,d)$ and $\cal P$ as the partition of ${\cal D}'$, and returns its output.
\end{proof}

Lastly, based on \Cref{thm:partialKernel} and \Cref{prop:polyequations}, we prove  \Cref{mainFPT} stating that \probRearr is FPT when parameterized by $d+k$. We restate the theorem here. 

\themainFPT*
%

\begin{proof}
	Given an instance  $({\cal S},k,d)$ of \probRearr, the algorithm first calls the algorithm in \Cref{thm:partialKernel} to obtain (in polynomial time) an equivalent instance  $({\cal S}',k,d)$ of \probRearr, where ${\cal S}'\subseteq{\cal S}$ is of size $\Oh(d^2k^3)$. Then, for every ${\cal A}\subseteq{\cal S}'$ of size at most $k$ such that ${\cal S}'\setminus{\cal A}$ is a packing, the algorithm tests whether it is possible to move each unit disk in $\cal A$ by a distance of at most $d$ so that, afterwards, ${\cal S}'$ becomes a packing. This can be done by using the algorithm in \Cref{prop:polyequations} to solve the following system of polynomial inequalities, which has variables $x_A,y_A$ for every $A\in{\cal A}$:
	\begin{itemize}
		\item For every $S\in{\cal S}'\setminus{\cal A}$ and $A\in{\cal A}$: $(x_A-a)^2+(y_A-b)^2\geq 4$, where $(a,b)$ denotes the center of $S$.
		\item For every distinct $A_1,A_2\in{\cal A}$: $(x_{A_1}-x_{A_2})^2+(y_{A_1}-y_{A_2})^2\geq 4$.
		\item For every $A\in{\cal A}$, where $(a,b)$ denotes the center of $A$ in ${\cal S}'$: $(x_A-a)^2+(y_A-b)^2\leq d^2$.
	\end{itemize}
	The correctness of the algorithm is immediate. For its running time analysis, notice that there are only $\sum_{i=0}^k{|{\cal S}'| \choose i}\leq (dk)^{\Oh(k)}$ choices for $\cal A$. Further, each of the systems of polynomial equations that are solved has at most $2k$ variables, degree $2$, and $\Oh(|{\cal A}|\cdot|{\cal S}'|)\leq (dk)^{\Oh(1)}$ equations. So, by \Cref{prop:polyequations}, it is solvable in time  $(dk)^{\Oh(k)} \cdot |I|^{\Oh(1)}$. In turn, we conclude that the algorithm runs in time $(dk)^{\Oh(k)} \cdot |I|^{\Oh(1)}$.
\end{proof}

\section{Kernelization lower bound for  \probRearr}\label{sec:rearr-nokern} 
In this section, we prove \Cref{thm:nopolyker}. To this end, we show that from several instances of \probCPack (defined below), we can construct a single instance $\cI'$ of \probRearr such that there is a solution to $\cI'$ if and only if there is a solution to at least one of the instances of \probCPack. The result then follows from the cross-composition technique (see \cite{FominLSZ19}, Chapter 17 for more details). \probCPack is defined as follows.

\defproblema{\probCPack}%
{A packing $\cP$ of $n$ unit disks inside a rectangle $R$ and an integer $\kappa \geq 0$.}%
{Decide whether there is a packing $\cP^*$ of $n+\kappa$ unit disks inside $R$ obtained from $\cP$ by adding $\kappa$ new disks.}
A recent result of Fomin et al. \cite{FominGISZ22,FominG0Z22ICALP} shows that the problem is \classNP-hard. In particular, they show the following result.

\begin{proposition}[Corollary 2 in \cite{FominGISZ22}]\label{cor:compl-hard-disk}
	\probCPack is \classNP-hard. Furthermore, it remains \classNP-hard, even when restricted to instances $(R, \cP, \kappa)$ of the following form.
	\begin{itemize}
		\item Rectangle $R$ is $[0, 2a] \times [0, 2b]$ for integers $a, b > 0$. It can also be assumed that $a = b$.
		\item A packing $\cP$ of disks with their centers inside $R$ such that (i) for every $i\in\{0,\ldots,a\}$, the disks with centers $(2i,0)$ and 
		$(2i,2b)$ are in $\cP$ and (ii) for   every $j\in\{0,\ldots,b\}$, the disks with centers $(0,2j)$ and 
		$(2a,2j)$ are in $\cP$.
	\end{itemize}
\end{proposition}

\begin{proof}[Proof of Theorem \ref{thm:nopolyker}.]
	The reader may wish to refer to \Cref{fig:nopolykernel}, which explains the schematics of the reduction. We consider instances $(R, \cP, n, \kappa)$ of \probCPack, where $R$ is an $[0,a] \times [0,a]$ square, where $a$ is an even positive integer, $\cP$ is a packing of $n$ disks with their centers inside $R$, such that the centers of the disks are rational, and $\kappa$ is the number of disks that need to be added inside $R$, which is compatible with $\cP$, to obtain a packing of $n+k$ disks. We also assume that  for every $i\in\{1,\ldots,a/2\}$, the disks with centers $(2i-1,1)$, $(2i-1,a-1)$, $(1,2i-1)$ and $(a-1,2i-1)$  are in $\cP$.
	
	For the cross-composition, we first show the polynomial equivalence relation $\mathcal{R}$, over instances $(R_i, \cP_i, n_i, \kappa_i)$ of \probCPack. The instances $(R_i, \cP_i, n_i, \kappa_i)$ and $(R_j, \cP_j, n_j, \kappa_j)$ go to the same equivalence classes if (1) the squares $R_i$ and $R_j$ have the same dimension, (2) $\cP_i$ and $\cP_j$ is a packing of $n_i = n_j$ disks inside $R_i$ and $R_j$ respectively with centers having rational coordinates, and (3) $\kappa_i = \kappa_j$. All the other malformed instances go into another equivalence class (see \cite{FominLSZ19} for the formal requirements of the equivalence relation). Note that $\mathcal{R}$ satisfies the properties of polynomial equivalence relation, since the equivalence can be checked in polynomial time, and (2) $\mathcal{R}$ partitions the elements of $S$ into at most $\lr{\max_{x \in S} |x|}^{\Oh(1)}$ classes in a well-formed instance, since $\kappa_i \le n_i$ (this can be assumed w.l.o.g.\ by padding the instance as required, as per 
	\Cref{cor:compl-hard-disk}). 
	
	Now we give a cross-composition algorithm for instances belonging to the same equivalence class. For the last equivalence class of malformed instances, we output a trivial no-instance. Thus, from now on, we focus on an equivalence class $(R_1, \cP_1, n, \kappa), \ldots, (R_t, \cP_t, n, \kappa)$, such that $a$ is the sidelength of every square $R_1, \ldots, R_t$. We assume w.l.o.g.\ that $t$ is odd, and $a$ is an even integer that is at least $10\kappa$.
	
	\begin{figure}
		\centering
		\includegraphics[scale=1]{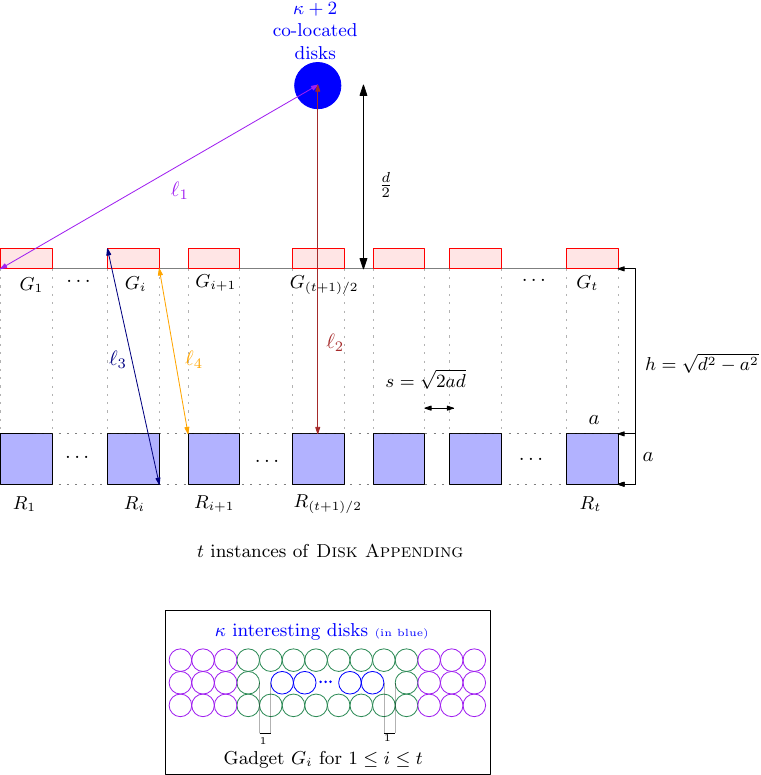}
		\caption{Schematic depiction of an instance of \probRearr obtained by OR-composition of instances $(R_i, \cP_i, n, \kappa)$ of \probCPack. Each instance $(R_i, \cP_i, n, \kappa)$ is shown in a blue square of sidelength $a$. Red rectangles are gadgets $G_i$, and an example gadget is shown below. Lengths $\ell_1, \ell_2, \ell_3, \ell_4$ are defined in \Cref{cl:nopolyker-distances}, and the values of $s$, $h$, and $d$ are carefully chosen functions of $t$ and $a$, in order to ensure that $\ell_1, \ell_3 \le d < \ell_2, \ell_4$ (note that the figure is not to scale). All the empty spaces are filled with padding disks with integral coordinates of centers. This ensures that an interesting disk from $G_i$ cannot be moved into an adjacent $R_{i\pm1}$, and thus different instances remain ``isolated''. In the gadget $G_i$, the \emph{surrounding disks} are shown in green and \emph{interesting disks} are shown in blue. Finally, purple disks are added on either side of the gadget in order to make the total width of the gadget exactly $a$. Then, the gadget $G_i$ and the corresponding square $R_i$ can be horizontally aligned as shown in the figure.} \label{fig:nopolykernel}
	\end{figure}
	
	For every $1 \le i \le t$, we construct a gadget $G_i$ as follows; see  \Cref{fig:nopolykernel}. Let $R$ be a rectangle of height $6$ and width $2\kappa+6$. Suppose the cartesian coordinates of the bottom-left corner of $G_i$ are $(0, 0)$ (note that this coordinate system is defined only for explaining the gadget structure, and should not be confused with the coordinate system in the next paragraph). Then, we place $2(k+3)$ disks centered at points $(1, 1), (3, 1), \ldots, (2k+5, 1)$, as well as $(1, 5), (3, 5), \ldots, (2\kappa+5, 5)$, and $2$ additional disks centered at $(1, 3)$, and $(2\kappa+5, 3)$. These disks lie along the perimeter of the rectangle, with centers at distance $1$ from the perimeter. We call these disks \emph{surrounding disks} (shown in green). Additionally, we place $\kappa$ disks with centers at $(4, 3), (6, 3), \ldots, (2\kappa+2, 3)$, which are termed as \emph{interesting disks} (shown in blue). Note that this leaves a horizontal gap of $1$ between the leftmost (resp.\ rightmost) interesting disk and the surrounding disks with center $(1, 3)$ (resp.\ $(2\kappa+5, 3)$). Now, we pad the gadget horizontally by adding columns of $3$ disks on both sides of the surrounding disks in a symmetric manner, such that the width of the gadget becomes exactly $a$. 
	
	Now we describe the construction of the instance of \probRearr. It might be useful to refer to a schematic description shown in \Cref{fig:nopolykernel}. Let $d$, the distance by which a disk can be moved, be equal to $\frac{9}{4}t^2a^2$. We place the first square $R_1$ and the corresponding packing of disks $\cP_1$ from the first instance by placing the bottom-left of $R_1$ corner at the origin $(0, 0)$. Next, we place the instances $(R_2, \cP_2), (R_3, \cP_3), \ldots, (R_t, \cP_t)$ by aligning their bottom edge along the $x$-axis, and leaving a horizontal gap of $s \coloneqq \sqrt{2ad}$ between the adjacent squares. Then, we place the gadgets $G_i$ directly above the rectangle $R_i$ such that the vertical distance between the top edge of $R_i$ and the top edge of $G_i$ is equal to $h \coloneqq \sqrt{d^2 - a^2}$. Since the width of every gadget $G_i$ is equal to $a$ after padding, the vertical boundaries of $R_i$ and the corresponding $G_i$ are aligned. Next, we place a set $C$ of $\kappa+2$ co-located disks such that (1) the vertical distance between the bottom edge of $G_{(t+1)/2}$ and the centers of disks in $C$ is equal to $d/2$, and (2) the centers of the disks in $C$ are aligned with the horizontal center of the gadget $G_i$. We place a rectangle tightly enclosing the instance constructed thus far, and pack all the empty spaces outside the gadgets using disks with integral coordinates on the centers (not shown in the figure). Finally, we set the budget $k$, the number of disks that can be moved, to be $(\kappa+1)+\kappa = 2\kappa+1$. This finishes the construction of the instance of \probRearr.

	The proof of the following claim follows from the careful choice of $d$, $h$ and $s$ in terms of $a$ and $t$.
	\begin{restatable}{claim}{distanceclaim} \label{cl:nopolyker-distances}
		\begin{enumerate}
			\item The maximum distance between the centers of disks in $C$ and any point in any $G_i$ is at most $d$ (shown as $\ell_1$ in \Cref{fig:nopolykernel}).
			\item The minimum distance between the centers of disks in $C$ and any point in any $R_i$ is more than $d$ ($\ell_2$).
			\item The maximum distance between a point in $G_i$ and a point in the corresponding $R_i$ is at most $d$ ($\ell_3$).
			\item The minimum distance between a point in $G_i$ and a point in another $R_j$ is more than $d$ ($\ell_4$).
		\end{enumerate}
	\end{restatable}
	\begin{proof}
		The values of $d, h$, and $s$ are chosen carefully in terms of $a$ and $t$ in order to ensure these properties. First we observe that $s = \sqrt{2ad} \ge a$, since $d = \frac{9}{4}t^2 a^2$.
		\begin{enumerate}
			\item Note that the horizontal distance between the midpoint of $G_{(t+1)/2}$ and the leftmost point in $G_1$ can be upper bounded by $(t/2)(a+s) \le (t/2) \cdot (2s) = t \cdot \sqrt{2ad}$.  The vertical distance between the bottom edge of $G_{(t+1)/2}$ and the centers of disks in $C$ is $d/2$. Therefore, it suffices to show that $\lr{\frac{d}{2}}^2 + \lr{t \sqrt{2ad}}^2 \le d^2$, i.e., $t^2 \cdot 2ad \le \frac{3d^2}{4}$, i.e., $2at^2 \le \frac{3}{4} \cdot \frac{9}{4} t^2 a^2$. This holds assuming $a \ge 216$.
			\item It suffices to consider the vertical distance between the centers of $C$ and the top edge of $R_{(t+1)/2}$. This vertical distance is $\frac{d}{2} + h - a $, which we want to show is greater than $d$. Note that it suffices to show that $h = \sqrt{d^2 - a^2} > \frac{d}{2}$, i.e., $a^2 < \frac{3d^2}{4}$, i.e., $\frac{243}{64} t^4 a^2 > 1$. However, since $t, a \ge 1$, this is true.
			\item $\ell_3^2 = h^2 + a^2 = d^2 - a^2 + a^2 = d^2$, since $h = \sqrt{d^2 - a^2}$.
			\item It suffices to consider adjacent $G_i$, $R_{i+1}$ pairs (argument for $R_{i-1}$ is identical). Then, $\ell_4^2 = (h-a)^2 + s^2 = (\sqrt{d^2 - a^2} - a)^2 + 2ad = d^2 - 2a\sqrt{d^2 - a^2} + 2ad $, which we want to show is at least $d^2$. This holds since $d > \sqrt{d^2 - a^2}$.
		\end{enumerate}
	\end{proof}

	Now we explain the implications of \Cref{cl:nopolyker-distances}. 
	In any yes-instance, at least $\kappa+1$ disks from $C$ must be moved by a distance at most $d$. Let $C'$ be this set of disks from the set of $\kappa+2$ co-located disks, that are moved. Note that in any gadget $G_i$, if all the $\kappa$ interesting disks are moved, then this creates an available space for placing $\kappa+1$ disks of $C'$. On the other hand, if any set of fewer than $\kappa$ disks inducing a connected component in the \emph{contact graph} (i.e., a special kind of intersection graph wherein there is an edge between the vertices corresponding to two disks iff their boundaries touch each other) of the disks is moved, then this creates space for at most $\kappa$ disks from $C'$ (note that the distance between $C'$ and an $R_i$ is more than $d$ by item 2 of \Cref{cl:nopolyker-distances}). However, since the budget is $2\kappa+1$, this cannot correspond to a feasible solution. Thus, in a solution to a yes-instance, $C'$ can only be moved in the place of $\kappa$ interesting disks corresponding to a gadget $G_i$. Next, an interesting disk can be moved anywhere in the corresponding square $R_i$ (item 3), but cannot be moved to a different square $R_j$ (item 4). Then, using an argument used for the disks in $C'$, we conclude that the $k$ interesting disks can only be moved in the empty spaces in the corresponding $R_i$. Thus, the created instance of \probRearr is a yes-instance iff there exists some yes-instance $(R_i, \cP_i, n, \kappa)$ of \probCPack. Finally, we note that \Cref{cor:compl-hard-disk} implies that the coordinates of the centers of the disks in each instance of \probRearr can be assumed to be rational. Furthermore, by letting $s \approx \sqrt{2ad} $, and $h  \approx \sqrt{d^2 - a^2}$ as rational approximations of their original values with small enough error, we can ensure that the coordinates of all the centers of the disks in the constructed instance become rational, and furthermore, the inequalities from \Cref{cl:nopolyker-distances} continue to hold. This concludes the proof of  \Cref{thm:nopolyker}.
\end{proof}


\section{\probRRearr}\label{sec:rearr-w-hard} 
In this section, we consider \probRRearr, which is the rectilinearly constrained version of \probRearr. We recall the definition of the problem for convenience.

\defproblema{\probRRearr}%
{A family $\cS$ of $n$ unit disks, an integer $k\geq 0$, and a real $d\geq 0$.}%
{Decide whether it is possible to obtain  from $\cS$ a family of non-overlapping disks  $\cP$ by moving at most $k$ disks into new positions parallel to the axes in such a way that each disk is moved at distance at most $d$.}

Note that \Cref{fig:introex2} is also  an example of \probRRearr, where the central disk is moved along the vertical axis. It can be easily verified that our algorithmic results, namely, \Cref{thm:partialKernel} and \Cref{cor:fpt} also hold for \probRRearr. On the other hand, as stated in \Cref{thm:w1hardness}, we show that \probRRearr is $\classW{1}$-hard parameterized by $k$. This section is dedicated to the proof of this theorem, which we restate below for convenience.

\themaiWhards*


\textbf{Proof of Theorem \ref{thm:w1hardness}: Overview.}
To prove that \probRRearr is \classW{1}-hard, we give a parameterized reduction from \probGridTiling, which is known to be \classW{1}-hard~\cite{Marx07a,CyganFKLMPPS15}.

\defproblema{\probGridTiling}%
{Positive integers $n, \kappa$, and a collection $\mathcal{S}$ of $\kappa^2$ non-empty sets $S_{i, j} \subseteq [n] \times [n]$ for $1 \le i, j \le \kappa$.}%
{Find integers $r^*_i \in [n]$ for all $1 \le i \le \kappa$, and $c^*_j \in [n]$ for all $1 \le j \le \kappa$, such that 
	for all $1 \le i, j \le \kappa$, $(r_i, c_j) \in S_{i, j}$.
}

\begin{figure}
	\centering
	\includegraphics[scale=0.8]{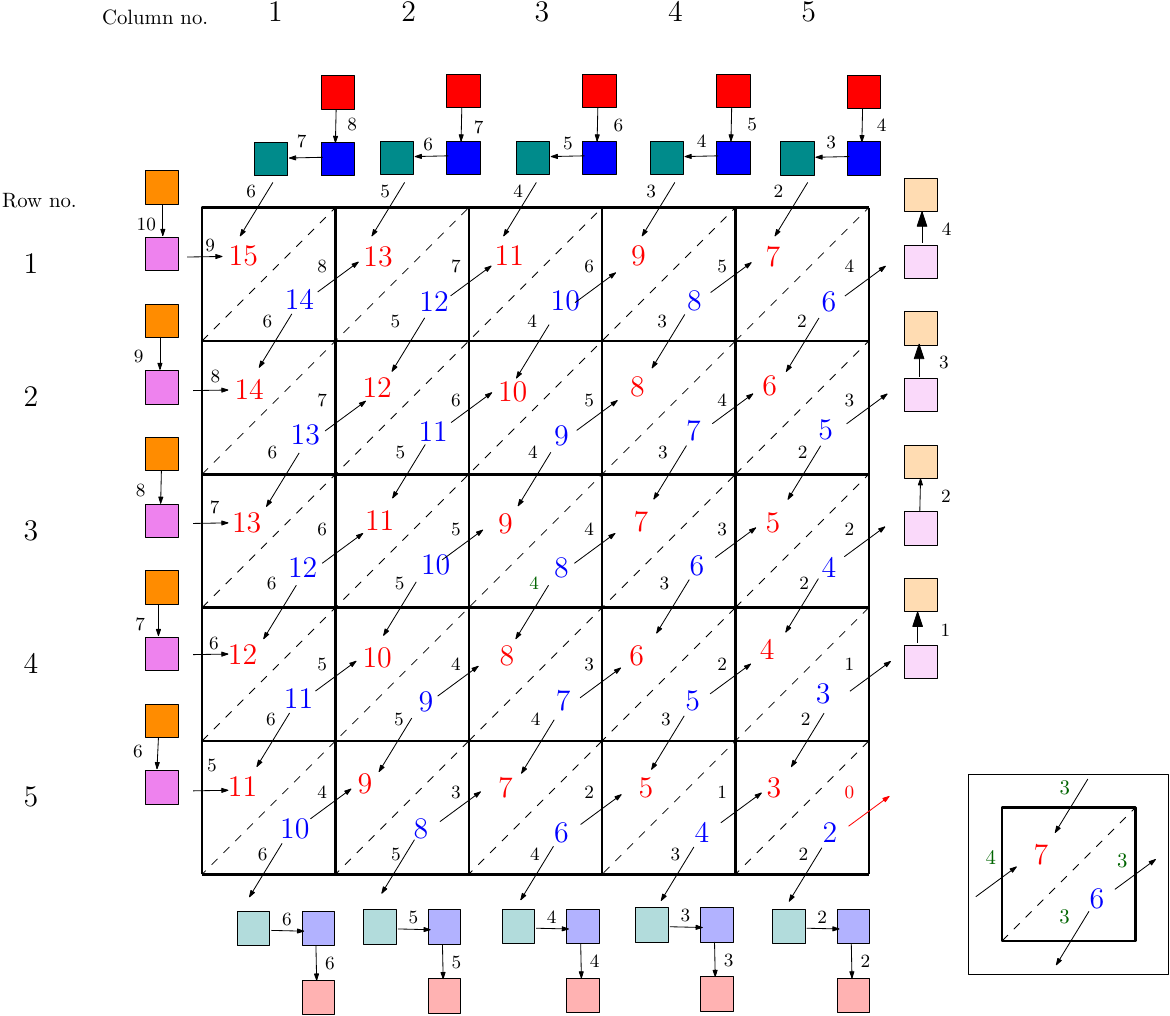}
	\caption{Left: A figure illustrating the skeleton of the instance produced by reduction from \probGridTiling. Right: An example cell $(3, 4)$.		
		The figure on the left represents the schematic of an instance produced via reduction of \probGridTiling with $\kappa = 5$ (note that figure is not to scale). The bulk of the instance is contained in a $5 \times 5$ grid as shown in the figure. Each grid cell corresponds to a cell gadget. There is a row gadget for every row (orange: $R^*(i)$ and purple: $RC(i, \cdot)$) to the left of the grid. Analogously, there is an emptying row gadget for every row except for $i = \kappa$ (faint orange $ER^*(i)$ and faint purple: $ERC(i, \cdot)$). Similarly, column gadgets above the grid and emptying column gadgets below the grid (red: $C^*(j)$, blue: $CC(j, n+1)$, and green: $CC(j, \cdot)$). The arrows between the gadgets represent the number of disks moved from one gadget to another in a yes-instance.
		To explain the various numbers written in each cell gadget, consider an example of a cell on the right, which corresponds to row $i = 3$ and column $j = 4$. Each cell gadget in $(3, 4)$ contains $6$ disks, denoted in blue. There are two incoming arrows into the cell $(3, 4)$, labeled with $3$ and $4$ respectively. This indicates that in a yes-instance, $3$ (resp. $4$) disks will move into one of the pair gadgets in $(3, 4)$ from the vertically (resp.\ horizontally) previous cell $(2, 4)$ (resp.\ $(3, 3)$). These $6$ disks in total (denoted in red) will displace $5$ disks in one of the cell gadgets in a yes-instance. Then, $3$ (resp.\ $3$) disks are moved into the vertically (resp.\ horizontally) next cell $(4, 4)$ (resp.\ $(3, 5)$). This is indicated with outgoing arrows and the respective labels.} \label{fig:gridex1}
\end{figure}

To give a high level overview of the reduction, we make some simplifications to aid the understanding. The idea is to ``embed'' the given instance of \probGridTiling in two dimensional plane using a set of unit disks (see \Cref{fig:gridex1}). Consider an $\kappa \times \kappa$ grid in the plane, which is divided into grid cells $(i, j) \in [\kappa] \times [\kappa]$. At the top of every column $1 \le j \le \kappa$, we create a \emph{column gadget}, containing a set of $c_j \ge 2$ co-located disks. Analogously, for every column $1 \le j \le \kappa$, we create an \emph{emptying column gadget} containing ``free space'' for $c'_j$ disks. However, since each disk can be moved by a distance of at most $d$, either horizontally or vertically, the co-located disks cannot be directly moved into the free spaces. We have an analogous construction of a row gadget to the left of every row $1 \le i \le \kappa$, containing $r_i$ co-located disks, and an emptying row gadget to the right of the row, containing free spaces for $r'_i$ disks. 

Each grid cell is further divided into sub-cells corresponding to pairs $(a, b) \subseteq [n] \times [n]$. For each pair $(a, b)$ that belongs to $S_{i, j}$, we create a \emph{pair gadget} $PG(a, b, i, j)$ in the corresponding sub-cell, containing $m_{ij}$ \emph{interesting disks}. However, this gadget has the property that there is a space for $m_{ij}+1$ disks, i.e., \emph{one additional disk} iff all the $m_{ij}$ interesting disks are moved elsewhere.

Note that any feasible solution must move at least $c_j - 1$ (resp.\ $r_i - 1$) of the co-located disks from every column gadget (resp.\ row gadget), in order to arrive at a non-intersecting configuration. This requires using the extra empty spaces in the pair gadgets $PG(\cdot, \cdot, \cdot, \cdot)$, as well as that in the emptying row and column gadgets. However, recall that a disk can only be moved horizontally or vertically by a distance of at most $d$. Thus, the co-located disks from the column gadget of column $j$ must be moved into a pair gadget corresponding to row $1$ and column $j$, say, $PG(a, b_j, 1, j)$, such that $(a, b_j) \in S_{1j}$. Due to the properties of the constructed instance, a feasible solution must make use of pair gadgets $PG(\cdot, b_j, i, j)$ for \emph{all rows} $1 \le i \le \kappa$, i.e., it enforces the choice of subcells of the form $(\cdot, b_j)$ for the column $j$. The argument for the row gadgets is analogous, which, for a row $1 \le i \le \kappa$, enforces the choice of subcells of the form $(a_i, \cdot)$ for all columns. 

If such consistent choices $1 \le a_i, b_j \le n$ exist for each row and column $1 \le i, j \le \kappa$, such that $(r_i, c_j) \in S_{i, j}$, then it is possible to move at most $k = f(\kappa)$ disks, either horizontally or vertically by a distance of at most $d$, to arrive at a non-intersecting configuration. On the other hand, if the given instance of \probGridTiling is a no-instance, then, the budget $k$ on the number of disks can be moved, is chosen in such a manner, that there exists no feasible solution that can achieve at a non-intersecting configuration by moving at most $k$ disks. This finishes the overview of the proof idea. Now, we turn to the formal proof.

\begin{proof}[Proof of Theorem \ref{thm:w1hardness}: Technical details.]
Now we discuss the details of the construction of the instance of \probRRearr obtained from \probGridTiling.


\medskip\noindent\textbf{Pair Gadget.} Consider some cell $(i, j)$, corresponding to row $1 \le i \le n$, and column $1 \le j \le n$. We create a pair gadget $PG(a, b, i, j)$ for every $(a, b) \in S_{i, j}$, see Figure \ref{fig:pair-gadget}. The exterior of a pair gadget is formed by unit disks arranged in a rectangular shape,. These disks are called \emph{Surrounding Disks} (shown in green). The width of a rectangle is $3$ unit disks, and the height is $L$ unit disks, for some large positive integer $L$, thus there are total $2L+2$ surrounding disks in total. For the sake of simplicity, let us assume that the cartesian coordinates of the center of the unit disk corresponding to the bottom-left (resp.\ top-right) of the rectangle is $(0, 0)$ (resp.\ $(2L+1, 4)$). There is an empty rectangle of dimension $(2L-2) \times 2$ inside the surrounding disks. We place two types of unit disks in this space: \emph{padding} (shown in red) and \emph{interesting} (shown in blue). The number of interesting disks in $PG(a, b, i, j)$ is equal to $m_{ij} \coloneqq 3\kappa-i-2j+2$. Note that the maximum value of $m_{ij}$ is equal to $3\kappa-1$ for $i = j = 1$, which we denote by $M$.

If column $j$ is odd, we place $N_1$ padding disks arranged vertically, starting from $(2, 2)$, where $N_1 \coloneqq \lfloor L/3 \rfloor$. If column $j$ is even, we place $N_1$ padding disks arranged vertically, starting from $(2, 3)$. We leave a vertical gap of $1$, and place $m_{ij}$ interesting blue disks above the $N_1$ padding disks (of type 1). Thus, the center of lowest interesting disk is given by $(2, 2 + 2N_1 + (j \mod 2))$. Then, we leave a vertical gap of $1$ above the highest interesting disk, and fill the remaining vertical space by padding disks (of type 2) touching each other vertically. Let the number of the type 2 disks be $N_2(i,j)$ -- this is a function of $i$ and $j$ because $L$ is fixed, whereas $m_{ij}$ is a function of $i$ and $j$. Note that if $j$ is odd, the highest type $2$ disks touches a surrounding disk; whereas if $j$ is even, there is a vertical gap of $1$.

Observe that there is a vertical gap of $1$ above and below the interesting disks in each pair gadget. Therefore, if all $m_{ij}$ interesting disks in $PG(a, b, i, j)$ are moved elsewhere, there is space for exactly $m_{ij}+1$ unit disks, although the centers of these potentially new disks will be vertically offset by $1$. Furthermore, the centers of interesting disks are vertically offset by a distance of $1$ for horizontally adjacent columns. Therefore, by moving distributing $m_{ij}$ disks into $PG(a, \cdot, i, j+1)$ and $PG(\cdot, b, i+1, j)$, and bringing in disks from $PG(a, \cdot, i, j-1)$, and $PG(\cdot, b, i-1, j)$, respectively, we gain space for exactly one additional disk inside $PG(a, b, i, j)$. This property is crucially used in the reduction.

\medskip\noindent\textbf{Absent Pair Gadget.} For all $(a', b') \not\in S_{i, j}$, we create an absent pair gadget $APG$ that has the same external dimension as a pair gadget, but is completely filled with surrounding and padding disks. Thus, an absent pair gadget consists of $L \times 3$ disks arranged in a rectangular grid. Note that the dimensions of a pair gadget as well as an absent pair gadget are $2L \times 6$.

\begin{figure}[H]
	\centering
	\includegraphics[scale=0.5]{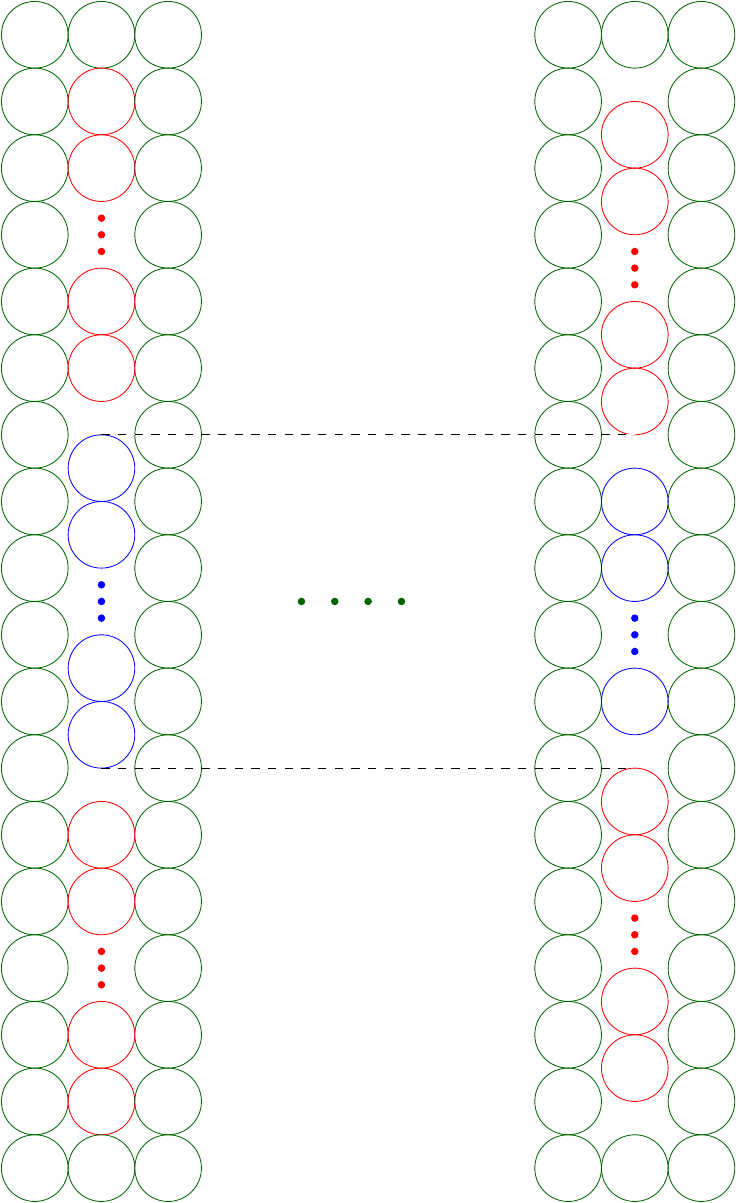}
	\caption{Two pair gadgets in horizontally adjacent cells, corresponding to an odd and even column. The disks in green are surrounding disks. There are padding disks (red) and interesting disks (blue) inside the contained formed by surrounding green disks. Note the vertical gaps of $1$ between padding and interesting disks, as well as additional vertical gaps of $1$ between padding and surrounding disks in the gadget corresponding to an even column on the right.}
	\label{fig:pair-gadget}
\end{figure}

Next we describe the construction of cell, row and column gadgets. For this, it will be useful to refer to Figure \ref{fig:gridex}. Note that the figure is not to scale -- in particular, the shapes of cell and grid gadgets are denoted by ``squares'', whereas our construction makes them into ``tall rectangles''. Furthermore, the large distances between cell/row/column gadgets are not represented to scale. Nevertheless, the figure should help the reader visualize the placement of various cell, row and column gadgets relative to each other.

\begin{figure}[H]
	\centering
	\vspace{-0.5cm}
		\includegraphics[scale=0.6]{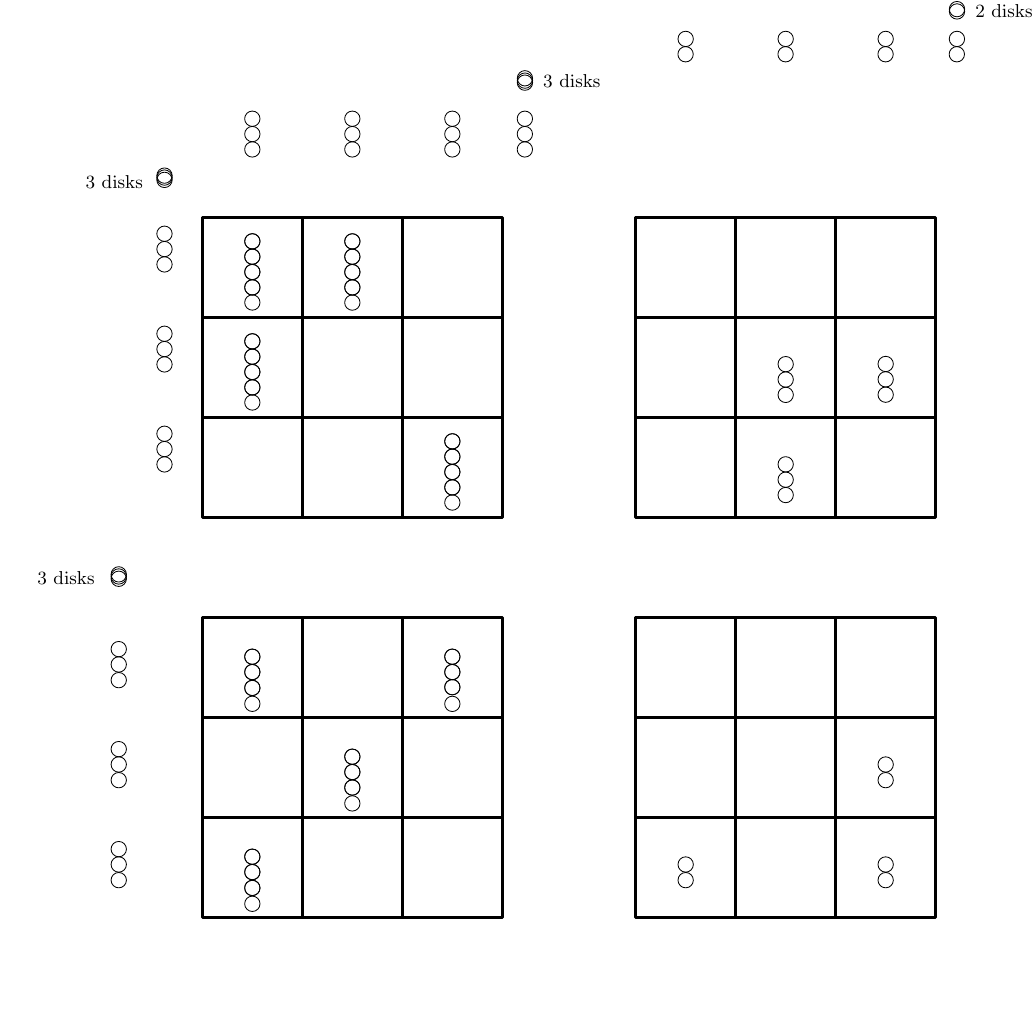}
	\hrule\ \\ 
	\vspace{0.5cm}
		\includegraphics[scale=0.6]{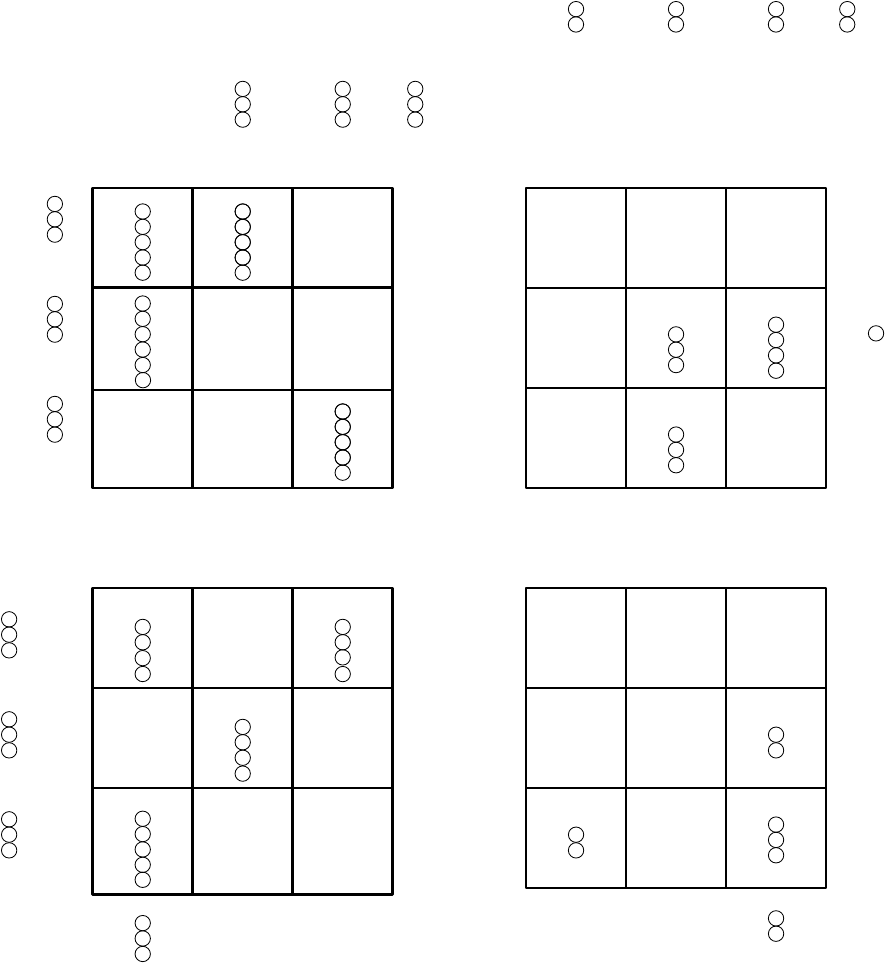}
	\caption{Consider the following instance of \probGridTiling. $n = 3$ and $\kappa = 2$, and let $S_{1, 1} = \{ (1, 1), (1, 2), (2, 1), (3, 3)\}, S_{1, 2} = \{ (2, 2), (2, 3), (3, 2) \}, S_{2, 1} = \{ (1, 1), (1, 3), (2, 2), (3, 1)\}, S_{2, 2} = \{ (2, 3), (3, 1), (3, 3) \}$. Note that this is a yes-instance with the following solution: $R_1 = 2, R_2 = 2, C_1 = 1, C_2 = 3$. The figure on the top shows the instance produced by the reduction. The figure on the bottom shows a corresponding solution after moving some of the disks rectilinearly by distance at most $d$, such that in the resulting configuration, no disks intersect. Note that the blank spaces are filled with ``irrelevant'' non-intersecting disks.}
	\label{fig:gridex}
\end{figure}

\medskip\noindent\textbf{Cell Gadget.} Consider a cell $(i, j)$ corresponding to a row $1 \le i \le \kappa$, and column $1 \le j \le \kappa$. We construct a cell gadget as follows. For each $(a, b) \in [n] \times [n]$, if $(a,b) \in S_{i, j}$, let $G(a, b, i, j) = PG(a, b, i, j)$ be the pair gadget, otherwise, let $G(a, b, i, j) = APG$ be the absent pair gadget.
Then, we arrange the gadgets $G(a, b, i, j)$ in an $n \times n$ rectangular grid, starting from the gadget $G(1, 1, i, j)$ in the top-left corner. Note that the dimensions of a cell gadget are $(2L n) \times (6n)$. 

\medskip\noindent\textbf{Arranging cells in a grid.} Now we arrange $\kappa \times \kappa$ cells in a grid in the following way. Starting from grid cell $(1, 1)$, corresponding to $i = j = 1$ is placed at the top-left corner. The cell gadgets of cells in the same row (resp.\ column) are aligned horizontally (resp.\ vertically), with the padding distance between the adjacent horizontal (resp.\ vertical) cell gadgets being $H$ (resp.\ $V$). As a result, for a fixed column $j$, all the pair gadgets of the form $PG(a, \cdot, i, j)$ are also aligned vertically for $1 \le i \le \kappa$, and analogously, for a fixed row $i$, all the pair gadgets of the form $PG(\cdot, b, i, j)$ are aligned horizontally for $1 \le j \le \kappa$. 

Now we specify the value of horizontal and vertical padding distances $H$ and $V$, and observe some properties. Let $V = 4nL$, and $H = 6nL - 12n > 0$, assuming $L$ is sufficiently large (the value of $L$ is later fixed to be $100 \cdot \max\{n, \kappa\}$). We set the distance $d$ by which any disk can be moved vertically or horizontally, to be $6nL$, i.e., $d \coloneqq 6nL$. Note that $d$ is \emph{not} a parameter of \probRRearr.

Now, notice that the vertical distance between a point in $PG(1, b, i, j)$ and that in $PG(n, b, i+1, j)$ is at most $2nL + V = 6nL = d$. On the other hand, the vertical distance between a point in $PG(n, b, i, j)$ and that in $PG(1, b, i+2, j)$ is at least $nL + 2V = nL + 8nL = 9nL > d$. Thus, a disk from a cell gadget can only be moved vertically into the cell gadget of vertically adjacent cells.

Similarly, the horizontal distance between a point in $PG(a, 1, i, j)$ and that in $PG(a, n, i, j+1)$ is at most $12n + H = 6nL = d$, whereas the horizontal distance between a point in $PG(a, n, i, j)$ and that in $PG(a, 1, i, j+2)$ is at most $6n + 2H = 12nL - 6n > d$, assuming $L$ is large enough. Thus, a disk from a cell gadget can only be moved horizontally into the cell gadget of horizontally adjacent cells.

\medskip\noindent\textbf{Row Gadget.} For a row $1 \le i \le n$, let $r_i \coloneqq 2\kappa-i$. The row gadget corresponding to row $i$ will consist of $n$ row-cell gadgets $RC(i, a)$ for $1 \le a \le n$. The structure of $RC(i, a)$ is the same as that of a pair gadget $PG$ corresponding to an even column -- it has dimensions $2L \times 6$, and contains surrounding, padding, and interesting disks respectively. The surrounding disks, and the padding disks of type $1$ are placed exactly as in a pair gadget $PG$. Then, we place $r_i$ interesting disks above the padding disks of type $1$ above the type $1$ padding disks, after vertical empty space of height $1$. The number of interesting disks in $RC(i, a)$ is equal to $r_i$. Then, we leave a vertical empty space of height $1$, and place an appropriate number of padding disks of type $2$. Note that similar to an even column, there is a vertical gap of $1$ at the top and bottom. The row gadgets $RC(i, a)$ are aligned with each other horizontally, and are placed at a horizontal distance of $4+ 6(i-1)$ to the left of the leftmost green \emph{surrounding} disks in the gadgets of the first column. Furthermore, the gadget $RC(i, a)$ is aligned vertically with the gadgets $PG(a, \cdot, i, \cdot)$.

Finally, we create an additional gadget $R^*(i)$ aligned horizontally with $RC(i, \cdot)$, of dimension $2L \times 6$. Like above, it contains $N_1$ padding disks of type $1$. Then, we place $r_i+2$ \emph{co-located} interesting disks above the top-most type $1$ padding disk, without leaving any vertical empty space. Above these $r_i+1$ interesting disks, we place an appropriate number of padding disks of type $2$, without leaving any vertical gap. Note that unlike $RC(i, a)$, there are no empty spaces of height $1$ inside $R^*(i)$. The gadget $R^*(i)$ is placed directly above $RC(i, 1)$. Note that the row gadgets corresponding to different rows are not aligned horizontally. This finishes the description of row gadget $R(i)$. 

Note that the vertical distance between a disk in $R^*(i)$, and that in $RC(i, n)$ is at most $2(n+1)L \le 6nL$. Therefore, the disks in $R_i$ can be moved vertically in the place of any of the disks in $RC(i, \cdot)$. The horizontal distance between a disk in $RC(i, a)$ and a point in $PG(a, n, i, 1)$ is at most $6n + 6\kappa + 6 \le d$, thus the disks in $RC(i, a)$ can be moved horizontally into any of the pair gadgets $PG(a, \cdot, i, 1)$. Finally, the horizontal distance between a disk in $RC(i, a)$ and a point in $PG(a, 1, i, 2)$ is at least $2H + 6n + 6 > d$, which implies that a disk in $RC(i, a)$ cannot be moved horizontally into a pair gadget of the second column. 

\medskip\noindent\textbf{Column Gadget.} The idea behind column gadgets is similar to that of row gadgets, with a few differences. Let $c_j = \kappa-j+2$. The column gadget corresponding to column $j$ will consist of $n$ column-cell gadgets $CC(j, b)$ for $1 \le b \le n$. The structure of a column cell gadget $CC(j, b)$ is exactly like a row cell gadget $RC(i, a)$, except that the number of interesting disks in $CC(j, b)$ is equal to $c_j$ The column gadgets $CC(j, b)$ are aligned vertically, and are placed at a vertical distance of $4 + (i-1)2L$ above of the topmost \emph{surrounding} disk in the gadgets of the first row. Furthermore, the gadget $CC(i, a)$ is aligned horizontally with the gadgets $PG(a, \cdot, i, \cdot)$.

Finally, we create two additional gadgets. The first is $CC(j, n+1)$, which resembles $CC(j, b)$ for $1 \le b \le n$, except that it contains $c_j+1$ interesting disks at appropriate locations. Furthermore, $CC(j, n+1)$ is vertically aligned with other gadgets $CC(j, b)$. The second gadget is $C^*(j)$, which is similar to $R^*(i)$, except the number of co-located disks is equal to $c_j+3$. This gadget is placed directly above $CC(j, n+1)$. This finishes the description of the column gadgets $C(j)$.

Note that the horizontal distance between a disk in $C^*(j)$ and a point in $CC(j, n+1)$ is at most $4L < d$, thus the disks in $C^*(j)$ can be moved vertically into $CC(j, n+1)$. Then, the horizontal distance between a point in $CC(j, n+1)$ and a point in $CC(j, 1)$ is at most $6(n+1) < d$, and thus the disks in $CC(j, n+1)$ can be moved into any of the $CC(j, \cdot)$. The vertical distance between a point in $CC(j, b)$ and a point in $PG(n, b, 1, j)$ is at most $2\kappa L + 2nL + 4 < d$, and thus disks in $CC(j, b)$ can be moved in the place of $PG(\cdot, b, 1, j)$. Finally, the vertical distance between a point in $CC(j, b)$ and a point in $PG(1, b, 2, j)$ is at least $2nL + V + 4 = 6nL + 4 > 6nL = d$. Therefore, a disk in $CC(j, b)$ cannot be moved into a pair gadget of the second row.

\begin{figure}
	\centering
	\includegraphics[scale=1.2]{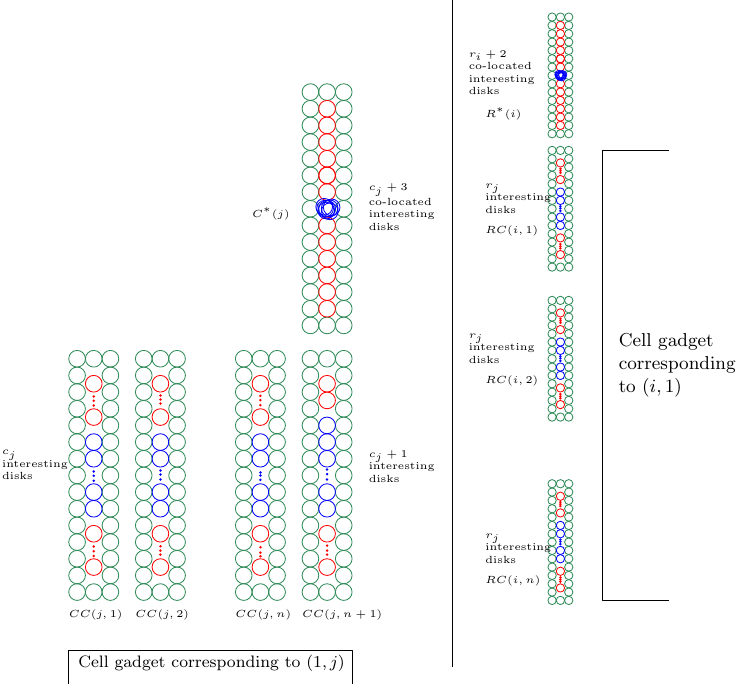}
	\caption{Left: Column gadget $C(j)$ corresponding to column $j$. Right: Row gadget corresponding to row $i$. Note that the vertical and horizontal gaps between sub-gadgets are only shown for aiding the visualization. There are no horizontal gaps between $CC(j, \cdot)$'s and no vertical gaps between $CC(j, n+1)$ and $C^*(j)$, as well as no vertical gaps between $RC(i, \cdot)$'s and $R^*(i)$. Also, the row gadget is compressed in size for compactness, but the dimension of sub-gadgets is the same as in column the gadget.}
\end{figure}

\medskip\noindent\textbf{Emptying Gadgets.} Corresponding to a row $1 \le i < \kappa$, we create an emptying row gadget $ER(i)$. Informally speaking, these gadgets are used to ``collect'' the leftover disks after a series of movements through the grid. Each of this gadget $ER(i)$, greatly resembles the row gadget as described above (note that for $i = \kappa$, we do not create $ER(\kappa)$). The emptying row gadgets $ER(i)$ are positioned symmetrically to the right side of the grid, mirroring the locations of the corresponding row gadgets $R(i)$ w.r.t.\ the grid. The external structure, and the dimensions of $ER(i)$ are same as that of $R(i)$. Furthermore, $ER(i)$, like $R(i)$, consists of emptying row-cell gadgets $ERC(i, a)$ for every $1 \le i \le a$, and an additional gadget $ER^*(i)$. 
Now we describe the internal differences between $ER(i)$ and $R(i)$.
\begin{itemize}
	\item The number of interesting disks in each $ERC(i, a)$ is equal to $er_i = \kappa-i$. The number of padding disks of type $2$, located above the interesting disks, is adjusted accordingly.
	\item The external structure, and the number of surrounding and padding disks in $ER^*(i)$ is the same as $ERC(i, \cdot)$. The only difference is that, instead of $\kappa-i$ interesting disks, it contains an empty space of height $2(\kappa-i)$ and width $2$. There is an appropriate number of padding disks of type $1$ and $2$, below and above the empty space, respectively.
\end{itemize}
Finally, corresponding to a column $1 \le j \le \kappa$, we create an emptying column gadget $EC(j)$, which resembles column gadget $C(j)$ (note that we create the gadget $CC(\kappa)$ for $i = \kappa$). The emptying column gadgets $EC(j)$ are positioned below the grid, and mirror the locations of the corresponding column gadgets $C(j)$ w.r.t.\ the grid. $EC(j)$, like $C(j)$, contains emptying column-cell gadgets $EC(j, b)$ for $1 \le b \le n$, and two additional gadgets $EC(j, n+1)$, and $EC^*(j)$. Just like $CC(j, \cdot)$, the number of interesting disks in $ER(j, b)$ is equal to $ec_j = \kappa-j+2$, for all $1 \le b \le n+1$, and the number of padding disks of type $2$ are adjusted accordingly. 

The only difference between $C(j)$ and $EC(j)$ is in the gadget $EC^*(j)$. The external structure, and the number of surrounding and padding disks in $EC^*(j)$ is the same as $ECC(j, \cdot)$. The only difference is that, instead of $\kappa-j+2$ interesting disks, it contains an empty space of height $2(\kappa-j+2)$ and width $2$. There is an appropriate number of padding disks of type $1$ and $2$, below and above the empty space, respectively.

Finally, we create a rectangle that precisely encloses all of the gadgets created above. Note that the coordinates of the vertices of the rectangle are integral. We pack all the empty spaces within the rectangle, but outside any of the created gadgets, using padding disks with centers with integer coordinates, such that the distance between the consecutive centers is $2$. Note that the dimensions of the rectangle are $O(n^3 \kappa) \times O(n^2 \kappa)$, thus, the size of the instance created is polynomial in the input size.

\medskip\noindent\textbf{Computing the value of the parameter $k$.} Recall that $k$ denotes the maximum number of disks that can be moved vertically or horizontally by the distance at most $d$. Now, the value of $k$ is defined to be
\begin{align*}
	\sum_{i = 1}^\kappa (r_i+r_i+1) + \sum_{j = 1}^\kappa (c_j + c_j + 1 + c_j+2) + \sum_{i = 1}^\kappa er_i + 2\sum_{j = 1}^\kappa ec_j + \sum_{i = 1}^\kappa \sum_{j = 1}^\kappa m_{ij}
\end{align*}
Recall that for any $1 \le i \le \kappa$, and $1 \le j \le \kappa$, the quantities in the summands are defined as follows: $r_i = 2\kappa-i$, $er_i = \kappa-i$, $c_j = ec_j = \kappa-j+2$, and $m_{ij} = 3\kappa-i-2j+2$. Note that each of the quantities in the summands is $O(\kappa)$, which implies that $k = O(\kappa^3)$. Finally, we fix the value of $L$ to be $100 \cdot \max \{n , k\}$. Note that this also fixes the value of $L$, and thus of the distance $d$. This finishes the construction of the reduced instance of \probRRearr. 

In the following discussion, we will explain how we arrive at the value of $k$, which will also demonstrate how a yes-instance of \probGridTiling corresponds to a yes-instance of \probRRearr. In this discussion, ``moving a \emph{disk} from gadget $A$ to gadget $B$'' always refers to moving an \emph{interesting disk} from gadget $A$ to gadget $B$, either vertically or horizontally by distance at most $d$. 

The first term corresponds to moving $r_i+1$ interesting disks from $R^*(i)$ to one of the $RC(i, a)$, and then $r_i$ disks from $RC(i, a)$ to some $PG(a, \cdot, i, 1)$, for every row $1 \le i \le \kappa$. Similarly, the second term corresponds to moving $c_j+1$ disks from $C^*(j)$ to $CC(j, n+1)$, from $CC(j, n+1)$ to some $CC(j, b)$, and then from $CC(j, b)$ to some $PG(\cdot, b, 1, j)$, for every column $1 \le j \le \kappa$. Note that there are $r_i+2$ (resp.\ $c_j+2$) mutually intersecting interesting disks in every $R^*(i)$ (resp.\ $C^*(j)$), and thus any solution has to move at least $r_i+1$ (resp.\ $c_j+1$) disks out of $R^*(i)$ (resp.\ $C^*(j)$).

The third term corresponds to moving $er_i$ disks from $ERC(i, a)$ to $ER^*(i)$ for every row $1 \le i \le \kappa$ (note that $er_1 = 0$). The fourth term corresponds to moving $ec_j$ disks from $ECC(j, b)$ to $ECC(j, n+1)$, and then from $ECC(j, n+1)$ to $EC^*(j)$, for every column $1 \le j \le \kappa$. Note that the gadgets $ER^*(i)$ and $EC^*(j)$ contain space for exactly $er_i$ and $ec_j$ interesting disks at the appropriate locations.

The last term corresponds to moving out $m_{ij} = 3\kappa-i-2j+2$ interesting disks out of one of the pair gadgets corresponding to cell $(i, j)$, say $PG(i, j, a, b)$. Out of these, $\kappa-j+2$ disks will be moved to a pair gadget $PG(i+1, j, \cdot, b)$ in the next row, and $2\kappa-i-j$ disks will be moved to a pair gadget $PG(i, j+1, a, \cdot)$ in the next column. If $i = \kappa$ (resp.\ $j = \kappa$), then $\kappa-j+2$ (resp.\ $\kappa-i$) disks will be moved to some emptying column gadget $EC(j, b)$ (resp.\ emptying row gadget $EC(i, a)$). 

Note that moving out all $m_j$ interesting disks creates a space for one additional disk between padding disks of type $1$ and type $2$. This space will be filled by incoming disks from previous gadgets, as follows. From a gadget $PG(i-1, j, \cdot, b)$, we move $\kappa-j+2$ disks into $PG(i, j, a, b)$, and from a gadget $PG(i, j-1, a, \cdot)$, we move $2\kappa-i-(j-1)$ disks into $PG(i, j, a, b)$. Note that the total number of incoming disks is $\kappa-j+2+2\kappa-i-(j-1) = 3\kappa-i-2j+3 = m_{ij}+1$. Thus, we use the additional space created. For $i = 1$ (resp.\ $j = 1$), the space in $PG(1, j, a, b)$ will be taken by $c_j = \kappa-j+2$ (resp.\ $r_i = 2\kappa-i$) disks in one of the column-cell (resp.\ row-cell) gadgets $CC(j, b)$ (resp. $RC(i, a)$). It is easy to verify that the number of incoming disks for $i = 1$ or $j = 1$ is also $m_{ij}+1$.

This discussion demonstrates that, if there exists a solution $r^*_i, c^*_j$ for $1 \le i, j \le n$, then it is possible to move exactly $k$ interesting disks either horizontally or vertically by a distance of at most $d$ between appropriate gadgets, to arrive at a configuration where no disks intersect. Thus, we have the following claim.

\begin{lemma}
	Assuming the original instance of \probGridTiling was a yes-instance, then the reduction produces a yes-instance of \probRRearr.
\end{lemma}

\medskip\noindent\textbf{No instances.} First, as discussed above, at least $r_i+1$ (resp.\ $c_j+2$) mutually intersecting interesting disks from every $R^*(i)$ have to be moved in any feasible solution. Note that the total number of these disks is $\sum_{i = 1}^\kappa 2\kappa-i+1 + \sum_{j = 1}^\kappa \kappa-j+4 = 2\kappa^2 + 4$. On the other hand, the total amount of empty space in the emptying row and column gadgets is equal to $\sum_{i = 1}^\kappa \kappa-i + \sum_{j = 1}^\kappa \kappa-j+2 = \kappa^2 -\kappa + 2$. Therefore, at a high level, we need to ``gain'' space for $\kappa^2 + \kappa + 2$ disks. It can be seen that, only way to ``gain'' space for one additional disk is to displace interesting disks from various gadgets. This is because, any surrounding or padding disk is placed in a configuration containing $\Omega(L) \gg k$ disks placed next to one another. Therefore, any solution that moves only at most $k$ disks cannot ``gain'' space by moving a surrounding or a padding disk.

Let $S_i$ be a set of $r_i+1$ disks from $R^*(i)$ that are moved (the argument for columns will be analogous). Note that all the eligible space, except that in the gadgets $RC(i, \cdot)$, is packed with padding or surrounding disks. However, since there are no gaps between such disks, if any subset of these disks inducing a connected component in the ``adjacency UDG'' containing at most $1 \le t \le r_i+1$ disks are moved, this will create space for exactly $t$ disks. However, the value of $k$ is chosen in such a way that, in order to arrive at a non-intersecting configuration, if $t$ disks are moved into new locations, it must displace at most $t-1$ other disks from their original locations; and furthermore, new locations of a set of disks cannot be a permutation of their old locations. For the sake of brevity, we refer to this argument as the \emph{equal displacement argument}.

Thus, the disks in $S_i$ can only be moved in one of the gadgets $RC(i, \cdot)$. Again, the disks in $S_i$ are aligned with the middle column of $RC(i, \cdot)$, consisting of a large number (i.e., ($\gg k$) of padding disks of types $1$ and $2$ respectively.  Therefore, the only feasible choice is to move $r_i+1$ interesting disks from one of the gadgets $RC(i, a)$ for some $1 \le a \le n$. Note that after displacing $r_i$ interesting disks in $RC(i, a)$, we ``gain'' two units of space, so that we can place $r_i+2$ disks of $S_i$. Note that we cannot ``split'' the disks in $S_i$ to be moved into different row cell gadgets $RC(i, \cdot)$, since we do not ``gain'' an extra space for one additional disk, which is required by the equal displacement argument. Arguing in this manner for the intersecting disks the row gadgets $R(i)$ for every $i$, and with appropriately modified argument for the column gadgets $C(j)$ for every $j$, we can show that exactly $r_i$ interesting disks are displaced from some row-cell gadget $RC(i, a)$, and exactly $c_j$ interesting disks are displaced from some column-cell gadget $CC(j, b)$. 

Again, using the equal displacement argument, we observe that the interesting disks from $RC(i, a)$ must be moved into some pair gadget $PG(i, 1, a, \cdot)$, and the interesting disks from $CC(j, b)$ must be moved into some pair gadget $PG(1, j, \cdot, b)$, displacing interesting disks therein. Note that the absent pair gadgets cannot be used for this purpose, since they contain a packing of disks, and equal displacement argument applies. Similarly, the interesting disks from $PG(i, j, a, b)$ can only be moved in the place of interesting disks in the pair gadgets of horizontally adjacent rows and columns. Since the emptying row and column gadgets are placed at the right and at the bottom respectively, and the way the budget $k$ is chosen, disks from the gadget of cell $(i, j)$ will not be moved into the previous row or into the previous column. Thus, it can be shown that if the original instance of \probGridTiling is a no instance, then there are no values $a_i, b_j$ for $1 \le i, j \le b$, such that interesting disks from $PG(i, j, a_i, b_j)$ displace interesting disks in $PG(i+1, j, a_{i+1}, b_j)$ and $PG(i, j+1, a_i, b_{j+1})$. However, since $\Theta(k^2)$ ``gains'' are required to arrive at a non-intersecting configuration, and each displacement of interesting disks corresponding to a pair gadget achieves a ``gain'' of $1$ disk, there is no solution to the created instance of \probRRearr.
This concludes the proof of the theorem.  
\end{proof}

%

\section{Conclusion and Open Problem}\label{sec:concl} 
In this paper, we initiate the study of the problem of spreading points from the perspective of parameterized complexity and kernelization. We reformulate the problem in terms of moving at most $k$ unit disks by a distance of at most $d$, which we call \probRearr. We design a (partial) polynomial kernel for \probRearr parameterized by $k$ and $d$. Furthermore, we show that this can be transformed into a (true) kernel, assuming the coordinates of the centers of the unit disks are rational numbers with bounded denominators. We complement this result by showing that \probRearr does not admit a polynomial kernel parameterized by $k$ alone, assuming $\classCoNP \subseteq \classNP/\poly$. These results provide a complete picture of the kernelization complexity of \probRearr.

We show that \probRearr is FPT parameterized by $k + d$, by combining the (partial) kernel with a non-trivial subroutine that involves solving a system of polynomial inequalities. It is natural to ask whether the problem is fixed-parameter tractable by the individual parameters $d$ and $k$. Fiala et al.\ \cite{FialaKP05} have shown that \probRearr is \classNP-hard even when $d  = 2$. 
Although the parameterized complexity of \probRearr parameterized by $k$ alone remains open, we make some preliminary progress in this direction, by proving that \probRRearr is  $\classW{1}$-hard when parameterized by $k$. 

\bibliographystyle{siam}
\bibliography{Disks}
\end{document}